\documentclass{amsart}

\usepackage{url}
\usepackage[english]{babel}
\usepackage[T1]{fontenc}
\usepackage{amsmath,amssymb,amsthm,amsfonts}
\usepackage{nicefrac}
\usepackage{overpic}
\usepackage{graphicx}
\usepackage{subfig}
\usepackage{algorithm}
\usepackage{algpseudocode} 
\usepackage[all]{xy}
\usepackage{tikz}
\usepackage{mathtools}
\usetikzlibrary{arrows,petri,topaths}
\usepackage{tkz-berge}

\algrenewcommand\algorithmicwhile{\textbf{While}}
\algrenewcommand\algorithmicfor{\textbf{For}}
\algrenewcommand\algorithmicdo{\textbf{Do}}
\algrenewcommand\algorithmicif{\textbf{If}}
\algrenewcommand\algorithmicthen{\textbf{Then}}
\algrenewcommand\algorithmicelse{\textbf{Else}}
\algrenewcommand\algorithmicend{\textbf{End}}
\algrenewcommand\algorithmicreturn{\textbf{Return}}

\theoremstyle{plain}
\newtheorem{lemma}{Lemma}[section]
\newtheorem{proposition}[lemma]{\textbf{Proposition}}
\newtheorem{theorem}[lemma]{\textbf{Theorem}}
\newtheorem{cor}[lemma]{\textbf{Corollary}}
\theoremstyle{definition}
\newtheorem{definition}[lemma]{\textbf{Definition}}
\newtheorem{example}[lemma]{\textbf{Example}}
\newtheorem*{notation}{\textbf{Notation}}
\newtheorem{remark}[lemma]{Remark}

\numberwithin{equation}{section}

\newcommand{\N}{\mathbb{N}}

\newcommand{\R}{\mathbb{R}}
\newcommand{\C}{\mathbb{C}}
\newcommand{\K}{\mathbb{K}}
\newcommand{\p}{\mathbb{P}}

\newcommand{\SE}[1]{\mathrm{SE}_{#1}}

\newcommand{\ppt}{\operatorname{pp}}
\newcommand{\spt}{\operatorname{sp}}
\renewcommand\tilde{\widetilde}
\renewcommand\hat{\widehat}
\newcommand{\tth}{\thinspace}

\newcommand{\rpos}[2]{\operatorname{RP}\!\left(#1,#2\right)}
\newcommand{\vrpos}[2]{\operatorname{VRP}\!\left(#1,#2\right)}
\newcommand{\rposconf}[3]{\operatorname{RP}\!\left(#1,#2,#3\right)}
\newcommand{\conf}[1]{\operatorname{Conf}(#1)}
\newcommand{\midpt}[1]{\chi(#1)}
\newcommand{\Flip}{\operatorname{Flip}}
\newcommand{\inv}{\operatorname{inv}}
\newcommand{\FM}[2]{\operatorname{FM}\!\left( #1, #2 \right)}
\newcommand{\IFM}[2]{\operatorname{IFM}\!\big( #1, #2 \big)}
\newcommand{\albar}{\overline{\alpha}}
\newcommand{\z}{\mathbf{z}}
\newcommand{\zt}{\tilde{\z}}
\newcommand{\cl}{\mathcal{C}\!\ell}

\title{Planar Linkages Following a Prescribed Motion}

\author[Gallet]{Matteo Gallet$^{\ast,\circ}$}
\author[Koutschan]{Christoph Koutschan$^\ast$}
\author[Li]{Zijia Li$^\ast$}
\author[Regensburger]{Georg Regensburger$^\dagger$}
\author[Schicho]{Josef Schicho}
\author[Villamizar]{Nelly Villamizar}

\thanks{$^\ast$ Supported by the Austrian Science Fund (FWF): W1214.} 

\thanks{$^\circ$ Supported by the Austrian Science Fund (FWF): P26607 - 
``Algebraic Methods in Kinematics: Motion Factorisation and Bond Theory''.}
\thanks{$^\dagger$ Supported by the Austrian Science Fund (FWF): P27229.}

\address[MG, CK, ZL, GR, NV]{
         Johann Radon Institute for Computational and Applied Mathematics (RICAM),
         Austrian Academy of Sciences,
         Altenberger Stra\ss e 69,
         4040 Linz, Austria}
\address[JS]{Research Institute for Symbolic Computation (RISC),
         Johannes Kepler University,
         Altenberger Stra\ss e 69,
         4040 Linz, Austria \vspace{.5cm}}
\email{\{matteo.gallet, christoph.koutschan, zijia.li, 
georg.regensburger,
\newline \phantom{mmmmmmmmmmmmmmmmm} 
 josef.schicho, nelly.villamizar\}@ricam.oeaw.ac.at}

\subjclass[2010]{Primary 70B15, 68W30, 70G55, 20G20, 16Z05, 14P05, 12Y05.\\
\phantom{Fn}
First published in Mathematics of Computation in 2016, published by American Mathematical Society}

\begin{document}

\begin{abstract}
Designing mechanical devices, called linkages, that draw a given plane curve 
has been a topic that interested engineers and mathematicians for hundreds of 
years, and recently also computer scientists. Already in 1876, Kempe proposed a 
procedure for solving the problem in full generality, but his constructions 
tend to be extremely complicated. We provide a novel algorithm that produces 
much simpler linkages, but works only for parametric curves. Our 
approach is to transform the problem into a factorization task over some 
noncommutative algebra. We show how to compute such a factorization, and how to 
use it to construct a linkage tracing a given curve.
\end{abstract}

\maketitle

\section{Introduction}
\label{introduction}

Kempe's Universality Theorem \cite{Kempe}, stating that any plane algebraic
curve can be drawn by a mechanical linkage with only rotational joints, 
surprised his contemporaries: his work represented a major breakthrough in a
topic investigated by  mathematicians of the 19th century. In the
middle of the last century, it appeared in textbooks, for example, the ones by
Lebesgue~\cite{Lebesgue1950} and Blaschke and M\"uller~\cite{BlaschkeMueller1956}.  
Recently, there was a revived interest in the problem among mathematicians and 
computer scientists as we outline in the following. We would also like to 
mention~\cite{Malkevitch2002} for a historical overview and further references.

In modern terms, the procedure proposed by Kempe is a parsing algorithm. It
takes the defining polynomial of a plane curve as input and realizes arithmetic
operations via certain elementary linkages. In this work, we approach the question from
a different perspective. Instead of a polynomial, we start with a rational
parametrization of a curve. From this, we obtain a parametrized family of
elements of the group~$\SE{2}$ of direct isometries of the plane, namely a
motion, whose action on a point traces the given curve. In order to realize such
a motion by a linkage, we decompose it into a series of revolutions. By encoding
motions via polynomials over a noncommutative algebra, we reduce this task to a
factorization problem. Eventually, we design a linkage whose rotational joints
move according to the previously obtained revolutions, yielding a device drawing
the desired curve.

Requiring a rational parametrization, it is clear that our approach is less
general than Kempe's, and its subsequent generalizations
\cite{Abbott2008, JordanSteiner1999, KapovichMillson2002, King1999}, 
because it cannot be applied to planar curves of
positive genus. However, the parametric setup has advantages in design
processes: in the language of robotics/kinematics, we are not just prescribing
the position of the end effector but also its orientation, in a similar way as
in~\cite{PlecMcWamp}; moreover,  when one has to interpolate prescribed poses 
of the end effector, we can specify a tracing order and have control over its 
speed. 

Linkages obtained via Kempe's procedure do not, in general, trace only the 
curve they are designed for. More precisely, the curve is only drawn by a
component of the configuration space (also called workspace), that is, the set 
of all positions reachable by the linkage. This is due to the fact that already the 
elementary linkages may admit degenerate configurations, allowing the device to 
flip into an unwanted component. Several solutions to this problem have been
suggested~\cite{Abbott2008,DemaineORourke2007,JordanSteiner1999,
KapovichMillson2002, King1999}. The linkages produced by our method present 
the same issue, however, this can be treated using the techniques presented
in~\cite{Abbott2008,DemaineORourke2007}; see Remark~\ref{remark:bracing}. 

Generation of linkages drawing an arbitrary curve may involve many links and
joints. This was already observed by Kempe~\cite{Kempe}:
\emph{``$\ldots$ this method would not be practically useful
on account of the complexity of the link work employed $\ldots$''}. The problem
of finding upper bounds for the number of joints needed to trace a given planar
curve of degree $d$ was first addressed in~\cite{GaoZhu1999,GaoZhuChouGe2001};
see~\cite[Section~3.2.3]{DemaineORourke2007}. Their 
upper bound $\mathrm{O}(d^4)$ was later improved by Abbott and Barton to
$\mathrm{O}(d^2)$, who in~\cite{Abbott2008} showed that this bound
is optimal. In this paper, we provide an algorithm that applies
to rational planar curves: given a rational parametrization with denominator of
degree~$d$ and without real roots, our algorithm produces a linkage with $3d+2$ 
links and $\frac{9}{2}d+1$ joints (Proposition~\ref{prop:draw}; we are
grateful to Hans-Peter Schr\"{o}cker for discussions leading to this result). If
we apply our technique to the parametrization of an ellipse, we obtain a linkage
with $8$ links and $10$ joints; in Section~\ref{overview}, we illustrate our 
approach and the main ideas of our paper via this example. In contrast, an 
unoptimized linkage returned by Kempe's procedure has $158$ links and $235$
joints; we are grateful to Alexander Kobel, who gave a full 
implementation~\cite{Kobel} of Kempe's procedure, for assisting us in the
computation of these numbers. 

In Section~\ref{linkages}, we develop a mathematical model for linkages with
only rotational joints. In particular, we define their configuration space in
terms of isometries sending a fixed initial configuration to a reachable one.
This differs from the commonly used models~\cite{GoodmanORourke2004,
JordanSteiner1999,KapovichMillson2002,King1999,OwenPower2009} which, instead, 
define configurations in terms of the positions of links and joints.

In Section~\ref{motion_polynomials}, we recall that one can embed $\SE{2}$ as 
an open subset of a real projective space; see~\cite{HustySchroecker}. This 
allows us to introduce a non\-commutative
algebra~$\K$ whose multiplication corresponds to the group
operation in~$\SE{2}$, hence mimicking the role played by dual quaternions with
respect to~$\SE{3}$. A polynomial
with coefficients in~$\K$ therefore describes a family of direct isometries,
which we call a rational motion. Consequently, we refer to such polynomials as motion polynomials; they are
the two-dimensional analog of the motion polynomials introduced 
in~\cite{HegedusSchichoSchroecker2013a}, and they are the key for turning our 
geometric problem into an algebraic one. In fact, we show that a linear motion
polynomial represents a motion constrained by a revolute or prismatic joint. 
Hence a factorization of a motion polynomial into linear factors yields a
decomposition of the corresponding rational motion into simple ones.

In Section~\ref{factorization}, we give necessary and sufficient criteria for
the existence of a factorization of a motion polynomial. Not every motion
polynomial admits such a factorization. However, the correspondence between
motion polynomials and rational motions is not one-to-one: for every rational
motion, there is a whole equivalence class of motion polynomials. Moreover, the
motions we are dealing with are special, since they admit bounded orbits. Once
we restrict to such bounded rational motions, we can prove that the equivalence
class of every bounded rational motion contains a motion polynomial admitting a
factorization into linear polynomials of revolute type. We provide an algorithm
for computing the smallest factorizable polynomial in an equivalence class and 
one of its factorizations; see Theorem~\ref{theorem:bounded_factor}. In this 
context, we would like to refer to the preprint~\cite{LiSchichoSchroecker2015},
where this result is used to prove factorizability of motion polynomials giving
rational motions in~$\SE{3}$.

At the beginning of Section~\ref{construction}, we construct from a 
factorization a linkage, in the form of an open chain, such that the given 
rational motion can be realized as the relative motion of the last link of the 
chain with respect to the first one. Open chains have high mobility, so we have 
to constrain our linkage such that it performs only the motion we are interested
in. The technique we employ, called flip procedure,  is introduced in 
Section~\ref{flip}. In the remaining part of Section~\ref{construction}, we 
exploit the properties of flips and propose our main algorithm; see 
Theorem~\ref{theorem:strong_realization}.

In Section~\ref{collisions}, we address the problem of self-collisions. For 
arbitrary linkages, this is a challenging problem \cite[Section 
9.3]{GoodmanORourke2004}. Here, for the first time in the paper, we take into 
account how linkages are physically realized. We show that  self-collisions can 
be efficiently detected for linkages obtained by our algorithm when links are 
realized by bars. If, instead, we allow links of different shapes, we describe 
a construction showing that it is possible to realize these linkages without 
collisions. This addresses the open problems raised in \cite[Open 
Problem~3.2]{DemaineORourke2007} and~\cite[Section~2.3]{ORourke2011}.

A popular formulation of Kempe's Theorem states that \emph{``There is a linkage 
that signs your name''}. Inspired by this, in Section~\ref{examples}, we give 
an example of a linkage drawing a calligraphic letter.

All algorithms described in the paper have been implemented by the second-named
author in the computer algebra system \textit{Mathematica}. The source code,
an expository notebook, and animations for the main examples are available for
free as electronic supplementary material~\cite{Koutschan15a}.

This paper is the outcome of the joint work of the Symbolic Computation group
at RICAM (Linz): the problem of constructing linkages following a planar
rational motion was proposed in our research seminar and after some time became
its main topic. Each participant contributed to the final result according to
his/her own background and skills (algebraic geometry, combinatorics,
kinematics, etc).

\section{A first example}
\label{overview}

We want to illustrate the main ideas and techniques with an example.
Consider the ellipse $(x+1)^2 + 4y^2 = 1$ in the plane,
which admits the rational parametrization
\[
  \varphi(t) \; = \; \left( \frac{-2}{t^2+1}, \frac{t}{t^2+1} \right).
\]
Our goal is to construct a planar linkage with rotational joints drawing the
curve parametrized by~$\varphi$ that admits only one degree of freedom.  This
means that there is a specific link (on which we put the pen) that performs a
motion along the ellipse as the linkage moves. By construction, the motion of
this link is the composition of several rotations.  We encode motions by
univariate polynomials with coefficients in the algebra~$\K= \C[\eta] /
(\eta^2, \imath \eta + \eta \imath)$, where $\imath$ denotes the imaginary
unit.  Elements of~$\K$ are of the form $z + \eta w$ with $z, w \in \C$, and
they are multiplied as follows:
\[
 (z + \eta w) \cdot (z' + \eta w') \; = \; (z \tth z') + \eta \tth (\overline{z} w' + z' w).
\]
In our case, a translational motion along the ellipse is represented by the polynomial
\[
  P(t) \; = \; (t^2 + 1) + \eta \tth (\imath t - 2).
\]
This means that the orbit of
any point under this motion is a translate of the ellipse. To construct a linkage that
realizes this motion, we want to factor~$P$ into linear polynomials, which correspond to
revolute motions. However, one can check that this is not at all possible!
On the other hand, as pointed out in Remark~\ref{remark:multiply_real}, we
will see that for any $R\in\R[t]$, the motion polynomial $RP \in \K[t]$
describes the same motion as~$P$. In the example, we can take $R=t^2+1$ so that
$RP$ admits a factorization into linear polynomials (see
Theorem~\ref{theorem:bounded_factor}):
\[
  R(t)\cdot P(t) \; = \; \bigl(t + \imath - \eta \, \imath \bigr) \cdot
  \bigl(t - \imath + \tfrac12 \eta \, \imath\bigr) \cdot
  \bigl(t - \imath + \tfrac32 \eta \, \imath\bigr) \cdot (t + \imath).
\]
Such a factorization is computed by the algorithm
\texttt{FactorMotionPolynomial} (see the end of
Section~\ref{factorization}). It allows us to construct a linkage, in the form
of an open chain (left part of Figure~\ref{figure:rigidification}), whose links
move according to the rotations represented by the linear factors (this is
algorithm \texttt{ConstructWeakLinkage} at the beginning of
Section~\ref{construction}). Since such a linkage has a high degree of
freedom, we need to constrain its mobility by adding more links and joints
(right part of Figure~\ref{figure:rigidification}). This is achieved by an
iteration of the so-called flip procedure, described in Sections~\ref{flip}
and~\ref{construction}.

\begin{figure}[ht]
\begin{center}
\begin{tabular}{cc}
\includegraphics[width=0.47\textwidth]{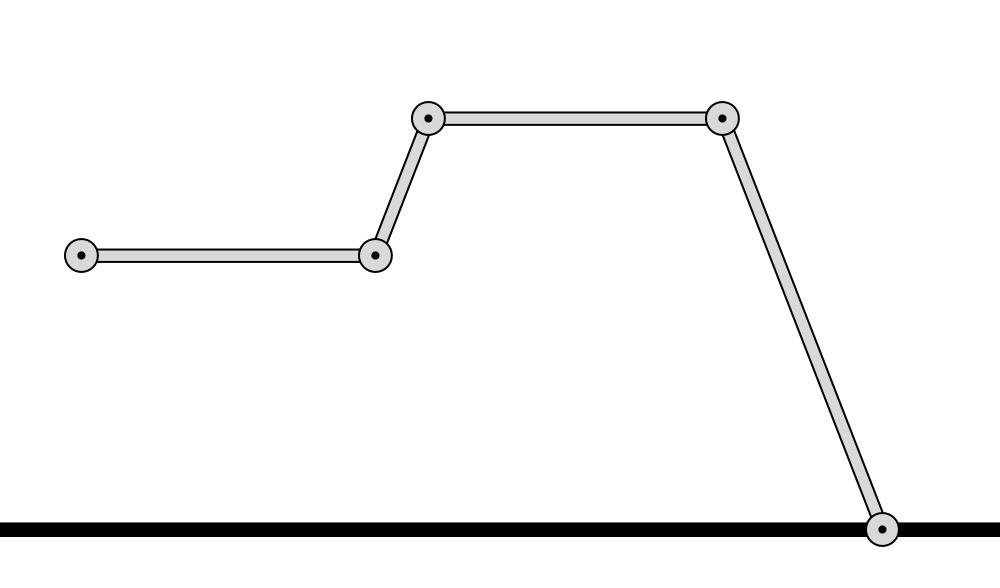} &
\includegraphics[width=0.47\textwidth]{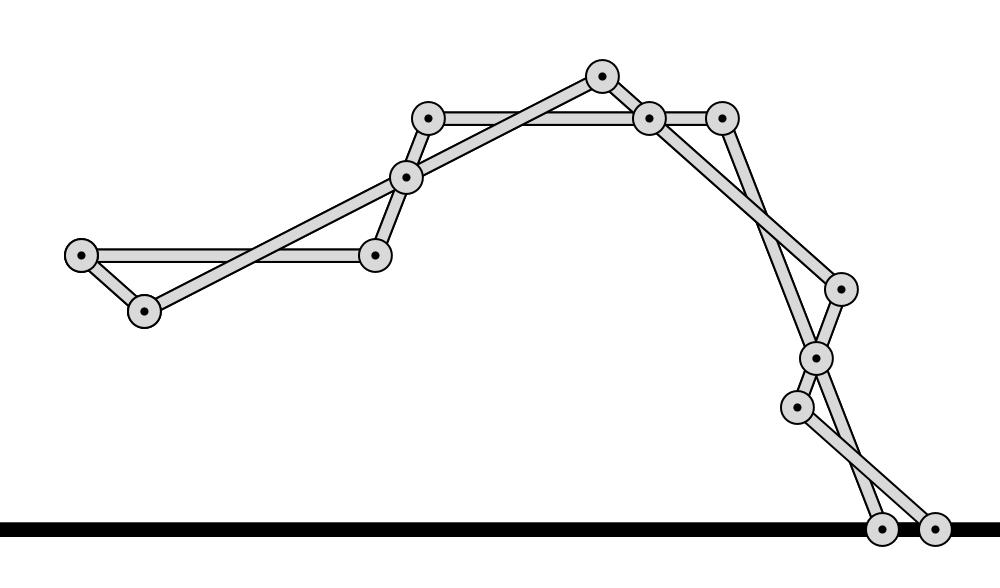}
\end{tabular}
\end{center}
\caption{An open chain and its extension to a linkage of mobility one,
realizing the translational motion given by $P(t)$.}
\label{figure:rigidification}
\end{figure}

However, if we just want to draw the ellipse, we need not realize exactly the
translational motion~$P(t)$: it is enough to find a motion for which the
orbit of one point is the ellipse. As pointed out in Remark~\ref{remark:multiply_complex},
multiplying with a polynomial $C\in\C[t]$ from the left does not change the
orbit of the origin. In our case, we find that the polynomial $CP$ with $C(t)=t-\imath$
factors completely as follows (see Proposition~\ref{prop:draw}):
\[
  C(t)\cdot P(t) \; = \; \bigl(t - \imath - \tfrac12 \eta\, \imath\bigr)\cdot
       \bigl(t - \imath + \tfrac12 \eta\, \imath\bigr)\cdot
       \bigl(t + \imath + \eta\, \imath\bigr).
\]
This factorization gives rise to a slightly simpler construction, see
Figure~\ref{figure:elliptic_linkage}. Further details concerning
this example will be provided in Examples~\ref{example:elliptic_motion},
\ref{example:elliptic_motion_reviewed}, and~\ref{example:elliptic_motion_flips}.

\begin{figure}[ht]
\includegraphics[width=0.70\textwidth]{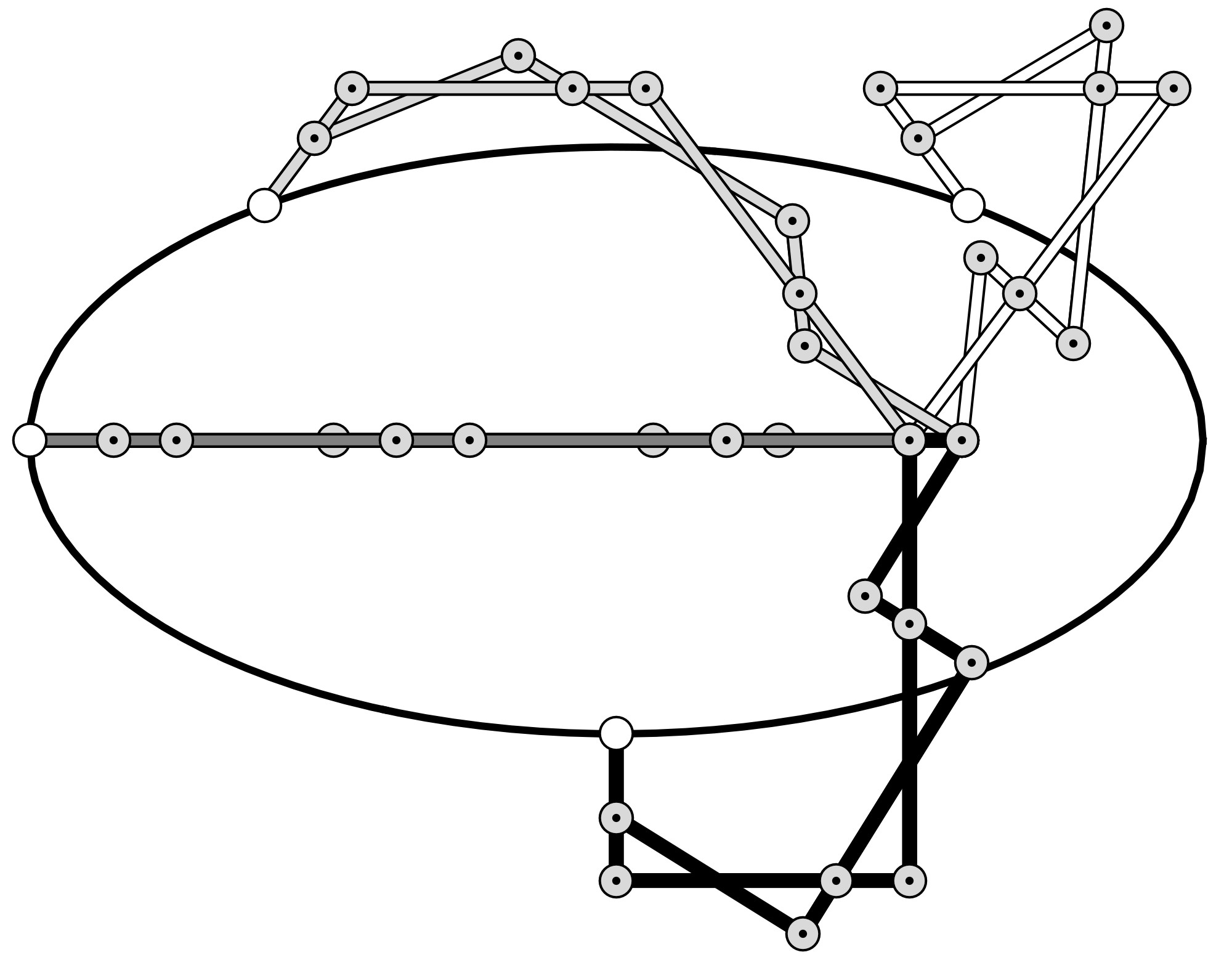}
\caption{The linkage that draws an ellipse. The same linkage is shown in
  different positions: $t=2$ (white), $t=\frac12$ (light gray), $t=0$ (dark
  gray), and $t=-1$ (black).}
\label{figure:elliptic_linkage}
\end{figure}

\section{Linkages}
\label{linkages}

In this section, we define a mathematical model for kinematic objects known as
\emph{linkages}. A linkage is a device constituted by rigid bodies, called
\emph{links}, connected by \emph{joints}, which restrict the relative position
of two neighboring links. In this paper we focus on \emph{planar} linkages
with \emph{revolute joints}, namely linkages for which all links move in
parallel planes and whose joints allow only rotations around a point. In our
model we will not be concerned with the shape of the links --- the only
exception will be Section~\ref{collisions} --- and we will suppose that joints
are the only constraints for the motion of the links.  Thus we can represent a
linkage by a graph whose vertices correspond to links, and where two vertices
are connected by an edge if and only if the corresponding links are connected
by a joint (for an example, see Figure~\ref{figure:graph_linkage}). But the
graph does not encode how the joints constrain the motion of the links, since
it does not specify their (initial) position. So we have to add this extra
information.

\begin{definition}
\label{definition:linkage}
A \emph{linkage} with revolute joints is a connected undirected graph $G = (V,
E)$, together with a map $\rho\colon E \longrightarrow \R^2$, such that $G$
does not have self-loops, i.e., all edges connect different vertices.  The
elements of~$V$ are called \emph{links}, while the elements of $E$ are called
\emph{joints}. We call two links \emph{neighboring} if they are connected by a
joint. For a joint $e \in E$, the point $\rho(e)$ is called the \emph{center
 of rotation} of~$e$.  In the following, we will always assume that $V$ is of
the form $\{1, \dotsc, n\}$ and that elements of~$E$ are given by unordered
pairs~$\{i,j\}$ of elements $i,j \in V$.
\end{definition}

\begin{remark}
An implementation of a linkage, where links are realized by line segments between
the joints, also looks like a graph. Note that this
graph is not the same as the one in Definition~\ref{definition:linkage}, rather it 
is its dual. The lengths of the links are then given implicitly by distances between
initial positions of joints.
\end{remark}
\begin{figure}[ht]
\begin{center}
  \includegraphics[width=0.5\textwidth]{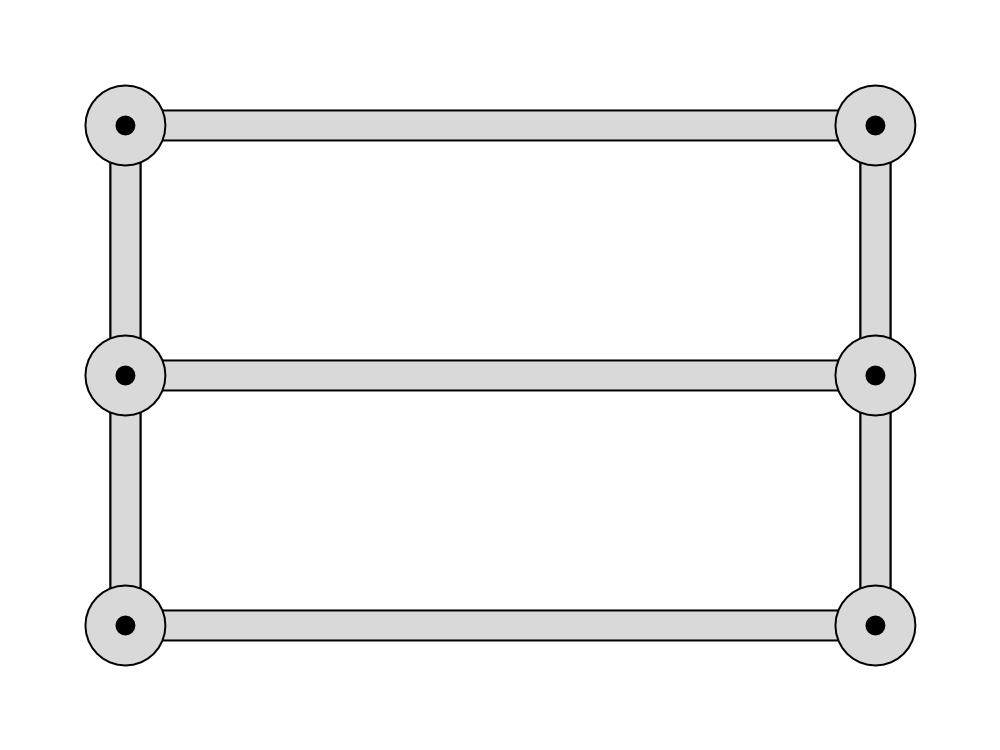}
  \hspace{1cm}
  \begin{minipage}[b]{0.4\textwidth}
  \begin{tikzpicture}
    \Vertex[x=0,y=1]{1}
    \Vertex[x=4,y=1]{2}
    \Vertex[x=2,y=2]{3}
    \Vertex[x=2,y=1]{4}
    \Vertex[x=2,y=0]{5}
    \Edge(1)(3)
    \Edge(1)(4)
    \Edge(1)(5)
    \Edge(2)(3)
    \Edge(2)(4)
    \Edge(2)(5)
  \end{tikzpicture}
  \vspace{1cm}
  \end{minipage}
\end{center}
\caption{A picture of a linkage (on the left) and the corresponding link graph
  (on the right).}
\label{figure:graph_linkage}
\end{figure}
Next we explain how we describe configurations of a linkage, namely, the
positions of all its links at a given moment. There are essentially two ways to 
indicate the position of a link: we can give either its absolute position with
respect to some frame of reference, or its relative position with
respect to some other link. Both these notions will be encoded by means of 
isometries.
\begin{notation}
We denote by~$\SE{2}$ the group of \emph{direct isometries} of~$\R^2$, i.e.,
maps that preserve distances and orientation (of the standard basis of~$\R^2$):
\[
  \SE{2} \; := \; \bigl\{ \sigma \colon \R^2 \longrightarrow \R^2 \, : \;
  \sigma \mathrm{\ is\ an\ isometry}, \; \det{\sigma} = 1 \bigr\}.
\]
\end{notation}

For each link~$i$ there is a unique isometry~$\sigma_i \in \SE2$ taking it
from its initial position to a current one; we say that $\sigma_i$ is the
\emph{absolute position} of the link~$i$. The \emph{relative position} of a
link~$i$ with respect to some other link~$j$ is defined to be~$\sigma_{i,j} =
\sigma_i \circ \sigma_j^{-1}$; hence we have the relation
$\sigma_{i,j}=\sigma_{i,k} \circ \sigma_{k,j}$. If $i$ and $j$ are neighboring
links, then $\sigma_{i,j}$ is a rotation around the point~$p=\rho(i,j)$, which
is the initial position of the joint $\{i,j\}$. This follows directly, since
$\sigma_i(p)=\sigma_j(p)$. All this motivates the following definition of
``virtual'' relative positions; the name indicates that they do not take into
account the constraints imposed by the linkage.

\begin{definition}
\label{definition:relative_position_neighboring}
	Let $L = (G, \rho)$ be a linkage, and let $i,j \in V$ be links connected by a 
joint~$e \in E$. The \emph{set of virtual relative positions} $\vrpos{i}{j}$ of 
the link~$i$ with respect to the link~$j$ is the subgroup of~$\SE{2}$ of 
rotations around the point~$\rho(e)$. Notice that $\vrpos{i}{j} = \vrpos{j}{i}$.
\end{definition}

We define a \emph{configuration} of a linkage~$L$ as a collection of virtual 
relative positions of links satisfying the following conditions: if $(i, h_1), 
(h_1, h_2), \dotsc, (h_{s}, i)$ is a \emph{directed cycle} in~$G$, namely a 
sequence of pairs of links connected by a joint, starting and ending at the same 
link, then  the composition $\sigma_{i, h_1} \circ \dotsb \circ \sigma_{h_s, 
i}$ --- which gives the relative position of~$i$ with respect to itself --- 
should be the identity. Loosely speaking, a configuration consists of the 
angles at all joints (counted twice), together with the (redundant) information 
$\rho(e)$ for all $e\in E$.

\begin{definition}
\label{definition:configuration}
The \emph{set of configurations} of a linkage~$L$ is defined to be
\begin{multline*} 
  \conf{L}  :=  \bigg\{ ( \sigma_{k,l} ) \in
 \!\!\! \prod_{\{i,j\} \in E} \vrpos{i}{j} \times \vrpos{j}{i} \, : \; \mathrm{for\ 
every\ directed\ cycle} \\
   (i, h_1), (h_1, h_2), \dotsc, (h_{s}, i) \mathrm{\ in\ } G, 
\mathrm{\ we\ have\ } \sigma_{i, h_1} \circ \dotsb \circ \sigma_{h_s, i} = 
\mathrm{id} \bigg\},
\end{multline*}
where $\prod$ denotes the Cartesian product so that $(\sigma_{k,l})$ is a tuple of size~$2|E|$.
\end{definition}

This definition of configuration space is the main difference between our model and the ones appearing in the literature~\cite{GoodmanORourke2004,
JordanSteiner1999,KapovichMillson2002,King1999,OwenPower2009}. 

\begin{remark}
\label{remark:projection_configuration}
	The cycle conditions imposed in the definition of $\conf{L}$ have the 
following consequence: suppose that $i,j \in V$ are two neighboring 
links; then $(i,j), (j,i)$ is a directed cycle from~$i$ to~$i$. Therefore, if 
$\Sigma = ( \sigma_{k,l} )$ is a configuration of~$L$, the 
cycle condition imposes that $\sigma_{i,j} \circ \sigma_{j,i} = \mathrm{id}$, 
implying that $\sigma_{i,j} = \sigma_{j,i}^{-1}$. Hence for every linkage $L = 
(G, \rho)$, the projection
\[ 
  \prod_{\{i,j\} \in E} \vrpos{i}{j} \times \vrpos{j}{i} \longrightarrow 
  \prod_{\substack{\{i,j\} \in E \\ i < j}} \vrpos{i}{j} 
\]
is a bijection when restricted to~$\conf{L}$ (and an isomorphism if we consider the projective structure
on~$\conf{L}$ we will define soon), and similarly for every projection that 
forgets exactly one among each pair~$(i,j)$ and~$(j,i)$. Still, we chose the 
above definition for $\conf{L}$ because it does not fix an orientation of the 
edges, allowing us to deal with arbitrary directed paths (see 
Definition~\ref{definition:relative_position}).
\end{remark}

Notice that for every pair  of neighboring links~$i$ and~$j$, the 
subgroup~$\vrpos{i}{j}$ can be set-theoretically identified with the real 
projective line~$\p^1_{\R}$. Under this identification, every cycle 
condition imposed in Definition~\ref{definition:configuration} becomes a closed 
condition in the Zariski topology, since it is given by multihomogeneous
polynomials. In this way $\conf{L}$ acquires the structure of a projective 
subvariety of~$\left( \p^1_{\R} \right)\!{}^{2 \left| E \right|}$.

\begin{definition}
\label{definition:mobility}
We define the \emph{mobility} of a linkage~$L$ to be the dimension 
of the configuration space~$\conf{L}$ as a projective subvariety of~$\left( 
\p^1_{\R}\right)\!{}^{2 \left| E \right|}$.
\end{definition}

So far we only took into account the relative position of two neighboring
links. For our purposes, namely to construct a linkage that follows a
prescribed motion, we need to take one link (the ``base'') as fixed, and
consider the relative positions\linebreak of all the other links with respect to the
base.

\begin{definition}
\label{definition:relative_position}
	Let $L = (G, \rho)$ be a linkage. Let $\Sigma \in \conf{L}$ and let $i, j \in
V$ be links. Let $(i, h_1), (h_1, h_2), \dotsc, (h_s, j)$ be a directed path 
in~$G$ from~$i$ to~$j$ --- which exists, since by 
Definition~\ref{definition:linkage} the graph~$G$ is connected. 
Then we define the \emph{relative position of~$j$ with respect to~$i$ in the 
configuration} $\Sigma=(\sigma_{k,l})$ as
\[
  \rposconf{i}{j}{\Sigma} \; := \; \sigma_{i, h_1} \circ \dotsb 
\circ \sigma_{h_s, j} \in \SE{2}.
\]
Notice that, because of the cycle condition, this definition is independent of 
the chosen path. We define the \emph{set of relative positions of~$j$ with 
respect to~$i$} to be
\[ \rpos{i}{j} \; := \; \bigl\{ \rposconf{i}{j}{\Sigma} \, : \; \Sigma \in
\conf{L} \bigr\} \; \subseteq \; \SE{2}. \]
\end{definition}

\section{Motion polynomials}
\label{motion_polynomials}
As mentioned in Section~\ref{introduction}, our goal is to reduce the main
problem to an algebraic one. We first introduce an algebraic setting for
manipulating isometries, consisting in a noncommutative $\R$-algebra~$\K$ whose 
multiplication corresponds to the group operation in~$\SE{2}$. This is an 
instance of a general construction associating a Clifford algebra 
to each group of isometries~$\SE{n}$; see for example
\cite[Section~9.2]{Selig2005}. Then we define the main object of our work,
motion polynomials, as polynomials over~$\K$. 

\smallskip
Following~\cite{HustySchroecker}, one can embed $\SE{2}$ in the real
projective space $\p^3_{\R}$ with coordinates $x_1,x_2,y_1,y_2$, as the open
subset
\begin{equation}
\label{equation:open_set}
  U \; := \; \p^3_{\R} \setminus \bigl\{ (x_1:x_2:y_1:y_2) \,:\; x_1^2 + x_2^2 =
0 \bigr\}.
\end{equation}
Hence $U$ is the complement of the projective line $x_1 = x_2 = 0$.  The 
(right) action of an element $(x_1:x_2:y_1:y_2)\in U$ on a point 
$(x,y)\in\R^2$ is given by:
\begin{equation}
\label{equation:action_SE2_R2}
  \begin{pmatrix} x \\ y \end{pmatrix}
  \;\; \mapsto \;\;
  \frac{1}{x_1^2 + x_2^2} \left[ 
  \begin{pmatrix} x_1^2 - x_2^2 & -2x_1 x_2 \\ 2x_1 x_2 & x_1^2 - x_2^2 
\end{pmatrix}
  \begin{pmatrix} x \\ y \end{pmatrix} +
  \begin{pmatrix} x_1 y_1 - x_2 y_2 \\ x_1 y_2 + x_2 y_1 \end{pmatrix} \right].
\end{equation}
Moreover, the group operation in $\SE{2}$ becomes a bilinear map:
\begin{equation}
\label{equation:mult_projective}
	\begin{aligned}
		(x_1: x_2: y_1: y_2) \cdot (x_1': x_2': y_1': y_2') \; = & \;\; \bigl( x_1 
x_1' - x_2 x_2' : \tth x_1 x_2' + x_2 x_1' : \\
		  & \;\; x_1 y_1' + x_2 y_2' + y_1 x_1' - y_2 x_2' : \\
		  & \;\; x_1 y_2' - x_2 y_1' + y_1 x_2' + y_2 x_1' \bigr),
	\end{aligned}
\end{equation}
where $(x_1: x_2: y_1: y_2)$ and $(x_1': x_2': y_1': y_2')$ represent
two direct isometries $\sigma, \sigma' \in \SE{2}$, respectively. 

For an easier handling of the multiplication in this
embedding of~$\SE{2}$, we introduce the following notation: we write a
representative $(x_1: x_2: y_1: y_2)$ of an element of $\SE{2}$ as a pair
$(z,w)\in\C^2$, where $z = x_1 + \imath \tth x_2$ and $w = y_1 + \imath \, y_2$ 
(here $\imath$ is the imaginary unit). Then
Equation~\eqref{equation:mult_projective} can be rewritten concisely as
\begin{equation}
\label{equation:mult_complex}
	(z,w) \cdot (z',w') \; = \; \bigl( z \tth z', \overline{z} \tth w' + z' \tth w
\bigr) 
\end{equation}
where $z, z', w$ and $w'$ are multiplied as complex numbers, and the bar 
$\overline{(\cdot)}$ is complex conjugation.
We can go further, by writing a pair $(z,w)$ in the form $z + \eta \tth w$, and
by postulating that $\eta$ satisfies the two relations:
\[
  z \tth \eta \; = \; \eta \tth \overline{z} \quad \text{for all } z \in \C
\quad 
  \mathrm{and} \quad \eta^2 \; = \; 0.
\]
Then the multiplication we obtain is exactly the one described in 
Equation~(\ref{equation:mult_complex}):
\[
	(z + \eta \tth w) \cdot (z' + \eta \tth w') \; = \; z \tth z' + \eta \tth
\bigl(\overline{z} \tth w' + z' \tth w \bigr).
\]
\begin{definition}
\label{def:algebraK}
We define the $\R$-algebra
\[
\K := \C[\eta]/(\eta^2,
\imath \tth \eta + \eta \tth \imath).
\]
Based on the previous discussions, we can identify elements of $\K$ with
elements of~$\SE{2}$. In this way the projective space
$\p^3_{\R}$ in which we embed $\SE{2}$ can be thought as the projectivization
$\p(\K)$ of $\K$, considered as an $\R$-vector space. 
\end{definition}

Notice that the algebra~$\K$ is constructed in such a way that its 
multiplication is a lift of the group operation of~$\SE{2}$: this means that 
if two isometries $\sigma_1,\sigma_2\in\SE{2}$ are represented by
$k_1,k_2 \in \K$, then $\sigma_2\circ\sigma_1$ is represented by~$k_1 \cdot 
k_2$ (remember we have a right action). More precisely, one can prove that 
$\K$ is isomorphic to the even subalgebra~$\cl^+(0,2,1)$ of the Clifford 
algebra~$\cl(0,2,1)$ of~$\SE{2}$; see~\cite[Section~9.2]{Selig2005}.

\begin{remark}
\label{remark:algebra_inverses}
  Looking at the construction of the algebra~$\K$, we notice that the condition
$x_1^2 + x_2^2 \neq 0$ for points in~$\p^3_{\R}$ representing isometries in
$\SE{2}$ becomes $z \neq 0$ when we consider elements of $\K$. Moreover, the
identity isometry is represented in~$\p^3_{\R}$ by the point $(1:0:0:0)$, hence
by any purely real element of $\K$, namely elements of the form $z + \eta \tth
w$ with $z \in \R$ and $w = 0$. From this we see that, given $k = z +
\eta \tth w \in \K$ representing an isometry $\sigma \in \SE{2}$, then $k' = 
\overline{z} - \eta \tth w$ represents~$\sigma^{-1}$, since the product~$k 
\tth k'$ equals~$|z|^2$, which is purely real. 
\end{remark}

The following lemma characterizes simple subgroups of~$\SE{2}$ in a 
geometric fashion: it shows that both translational and revolute motions,
i.e., families of translations resp.\ rotations, correspond to lines
in~$\p^3_{\R}$.
\begin{lemma}
\label{lemma:one_dim_motions}
Let $\ell \subseteq \p^3_{\R}$ be a projective line passing through the point
$(1:0:0:0)$, and define $\ell_{U} = \ell \cap U$ and
$X = \ell \setminus U$, where $U$ is defined 
in Equation~\ref{equation:open_set}. Then:
\begin{enumerate}
\item\label{lemma:one_dim_motions:case1}
  if $X$ has cardinality~$1$, then $\ell_{U}$ corresponds to a
  subgroup of~$\SE{2}$ that consists of all translations along a fixed common
  direction;
\item\label{lemma:one_dim_motions:case2}
  if $X$ is empty, then $\ell_{U}$ corresponds to a subgroup
  of~$\SE{2}$ that consists of all rotations around a fixed common point.
\end{enumerate}
\end{lemma}
\begin{proof}
We analyze the two cases separately.\\
\textit{Ad \ref{lemma:one_dim_motions:case1}:}
By hypothesis we have that $(1:0:0:0) \in \ell$ and $(0:0:a:b) \in X\subset\ell$
for some $a,b \in \R$, not both zero. Hence the line~$\ell$ admits the 
parametrization:
\[ (\lambda: 0: a \mu: b \mu), \qquad \text{for\ } (\lambda: \mu) \in \p^1_{\R}.
\]
Plugging this parametrization into Equation~\eqref{equation:action_SE2_R2},
we see that the elements of~$\ell_{U}$ are translations by the 
vector $\nicefrac{\mu}{\lambda}\,(a,b)$.\\
\textit{Ad \ref{lemma:one_dim_motions:case2}:}
Again by hypothesis we have that $(1:0:0:0) \in \ell$; moreover there
exists a point in~$\ell$ which is of the form $(0:a:b:c)$. Since $X$ is
assumed to be empty, we have $a\neq0$, and hence the line~$\ell$ admits
the parametrization:
\[
  (\lambda : a \mu: b \mu: c \mu), \qquad \text{for\ } (\lambda: \mu) \in
\p^1_{\R}.
\]
We use Equation~\eqref{equation:action_SE2_R2} to compute the fixed points of 
an arbitrary element of~$\ell$. A direct calculation reveals that we always get 
the point $(-\nicefrac{c}{2a},\nicefrac{b}{2a})$, which is independent 
of~$\lambda$ and~$\mu$. Thus every point of~$\ell_{U}$ represents a 
rotation around~it.
\qedhere
\end{proof}

The description of revolutions around a point given by 
Lemma~\ref{lemma:one_dim_motions} enables us to compute effectively the configuration 
curve of a linkage, as shown in Example~\ref{example:mobility}.

\begin{example}
\label{example:mobility}
Let us consider the linkage $L$ whose graph is depicted in 
Figure~\ref{figure:mobility}. The linkage~$L$ is given by $4$ links and $4$ 
joints, and the map~$\rho$ is determined by:
\[
  \begin{array}{ccccc}
    \rho(1,3) & = & u_1 & = & \left( 0, -\nicefrac{3}{2} \right), \\
    \rho(3,4) & = & u_3 & = & \left( -\nicefrac{3}{4}, \nicefrac{9}{2} \right), 
\\
    \rho(2,4) & = & u_4 & = & \left( -1, \nicefrac{13}{2} \right), \\
    \rho(1,2) & = & u_2 & = & \left( -\nicefrac{1}{4}, \nicefrac{1}{2} \right).
  \end{array}
\]
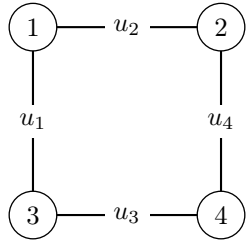
\begin{figure}
\begin{tikzpicture}
  \Vertex[x=0,y=2.5]{1}
  \Vertex[x=0,y=0]{3}
  \Vertex[x=2.5,y=0]{4}
  \Vertex[x=2.5,y=2.5]{2}
  \Edge[label=$u_1$](1)(3)
  \Edge[label=$u_3$](3)(4)
  \Edge[label=$u_4$](4)(2)
  \Edge[label=$u_2$](1)(2)
\end{tikzpicture}
\caption{Graph of a linkage constituted by $4$ links and $4$ joints, a 
so-called closed $4R$-linkage. The edges are labeled by points~$u_i$ in~$\R^2$ 
that are the images of the edges under the map~$\rho$.}
\label{figure:mobility}
\end{figure}
\noindent We compute the configuration set of~$L$ and its dimension as a 
projective variety. Since we have $4$ joints, it will be a projective 
subvariety of $\left( \p^1_{\R} \right)\!{}^4$ (here, rather than $\conf{L}$, 
we are considering one of its projections, but as noticed in 
Remark~\ref{remark:projection_configuration} we get an isomorphic object). 
Lemma~\ref{lemma:one_dim_motions} ensures that the $4$ subgroups of~$\SE{2}$ 
of rotations around the points~$u_i$ --- which constitute the
subgroups of virtual positions between two neighboring links as in
Definition~\ref{definition:relative_position_neighboring} --- correspond to $4$
lines $\ell_i \subseteq \p^3_{\R}$ with parametrizations:
\[
  \begin{array}{ccc}
  \textrm{center } u_i & \qquad & \textrm{parametrization of line } \ell_i \\
  \hline
  \left( 0, -\nicefrac{3}{2} \right) & & (\lambda_1: \mu_1: -3 \mu_1: 0) \\
  \left( -\nicefrac{3}{4}, \nicefrac{9}{2} \right) & & (\lambda_3: 2 \mu_3: 18 
\mu_3: 3 \mu_3)\\
  \left( -1, \nicefrac{13}{2} \right) & & (\lambda_4: \mu_4: 13 \mu_4: 2 \mu_4) 
 \\
  \left( -\nicefrac{1}{4}, \nicefrac{1}{2} \right) & & (\lambda_2: 2\mu_2: 
2\mu_2: \mu_2 )
  \end{array}
\]
where $(\lambda_i: \mu_i) \in \p^1_{\R}$ for all $i \in \{1, \dotsc, 
4\}$. Recall from Definition~\ref{definition:configuration} that, in this case, 
a configuration for~$L$ is a $4$-tuple $(\sigma_{1,2}, \sigma_{2,3}, 
\sigma_{3,4}, \sigma_{4,1})$ of direct isometries that satisfies the cycle 
condition $\sigma_{1,2} \circ \sigma_{2,3} \circ \sigma_{3,4} \circ \sigma_{4,1} 
= \mathrm{id}$. Each of the isometries~$\sigma_{k,l}$ gives a point on one
projective line~$\ell_i$, and in our projective model of~$\SE{2}$ composition
corresponds to multiplication according to
Equation~\eqref{equation:mult_projective}. Under 
these identifications, the composition $\sigma_{1,2} \circ \sigma_{2,3} \circ 
\sigma_{3,4} \circ \sigma_{4,1}$ corresponds to a point $(F_1: F_2: F_3: F_4) \in 
\p^3_{\R}$, where all $F_i$ are polynomials in the variables 
$(\lambda_1:\mu_1), \dotsc, (\lambda_4:\mu_4)$. Noticing that the identity 
element of~$\SE{2}$ is represented by the point $(1:0:0:0) \in \p^3_{\R}$, one 
realizes that the cycle condition is equivalent to
\[
  \operatorname{rk}\begin{pmatrix}
  F_1 & F_2 & F_3 & F_4 \\ 1 & 0 & 0 & 0
  \end{pmatrix} \; = \; 1.
\]
Hence, as a subvariety of $\left( \p^1_{\R} \right)\!{}^4$, the configuration 
set~$\conf{L}$ is the zero set of the polynomials $F_2, F_3, F_4$. A computer 
algebra computation shows that this is a one-dimensional variety with two 
components.
\end{example}

We introduce now one of the main concepts of this paper, namely 
\emph{motion polynomials}. Intuitively, a motion can be described as a curve in 
the space of direct isometries. In our case, we want such a curve to be defined 
by a rational parametrization.

\begin{definition}
\label{definition:rational_motion}
  Let $X_1, X_2, Y_1, Y_2 \in \R[t]$ be polynomials such that $X_1^2 +
X_2^2$ is not identically zero; denote by $V(X_1, X_2, Y_1, Y_2)$ the
set of their common zeros in~$\R$. The map $\phi \colon \R \setminus V(X_1, X_2,
Y_1, Y_2) \longrightarrow \p^3_{\R}$ defined by $X_1, X_2, Y_1, Y_2$ is called a
\emph{rational motion}. We will use the notation
$\phi \colon \R \dashrightarrow \p^3_{\R}$ to mean that $\phi$ is not
defined everywhere on~$\R$. Thinking of $\SE{2}$ as an open subset
of~$\p^3_{\R}$, the condition on $X_1^2 + X_2^2$ ensures that for all but
finitely many $t \in \R$ we have $\phi(t) \in \SE{2}$.
\end{definition}
\begin{definition}
\label{definition:motion_poly}
	Let $\phi$ be a rational motion given by $(X_1,X_2, Y_1, Y_2)$. The polynomial
$P(t) = Z(t) + \eta \tth W(t) \in \K[t]$, where $Z= X_1 + \imath X_2$ and $W =
Y_1 + \imath Y_2$, is called a \emph{motion polynomial} encoding the
motion~$\phi$. The polynomials~$Z$ and~$W$ are respectively called the
\emph{primal} and \emph{secondary part} of~$P$, denoted by~$\ppt(P)$
resp.~$\spt(P)$.
\end{definition}

Next we want to connect motion polynomials with rational plane curves.
For this purpose it is helpful to rephrase the action of~$\SE{2}$ on~$\R^2$
shown in Equation~\eqref{equation:action_SE2_R2} in the 
following way: we identify a point $(x,y) \in \R^2$ with the 
element $u = x + \imath y \in \C$; if $\sigma \in \SE{2}$ is represented by 
an element $z + \eta \tth w \in \K$, then $\sigma$ sends~$u$ to
\begin{equation}\label{equation:action_K_C}
  \frac{uz^2 + zw}{|z|^2}.
\end{equation}
Using this formulation of the (right) action, it is easy to prove the following 
result.
\begin{proposition}
\label{prop:curve_motion}
Let $\varphi\colon \R \longrightarrow \R^2$ be a rational parametrization of a 
real curve, which means that $\varphi$ is of the form
\[
  \varphi(t) \; = \; \left( \frac{f(t)}{h(t)}, \frac{g(t)}{h(t)} \right)
\]
for some real polynomials~$f,g$ and~$h$. Then the orbit of the origin
under the motion given by the motion polynomial 
$P \; = \; h + \eta \tth (f + \imath g)$
is exactly the image of~$\varphi$.
\end{proposition}

Let $P_1$ and $P_2$ be two motion polynomials. Because of the algebraic
properties of~$\K$, the polynomial~$P = P_1 P_2$ defines a motion that is the
composition of the motions determined by~$P_2$ and~$P_1$. More precisely, for 
every $t \in \R$, the isometry~$P(t)$ is the composition of the 
isometries~$P_2(t)$ and~$P_1(t)$. This shows how important the factorization of 
a motion polynomial is in our framework: it provides a decomposition of a 
motion into simpler ones.

We saw  in Lemma~\ref{lemma:one_dim_motions} that rational motions whose image 
is a line in~$\p^3_{\R}$ are translational motions or revolutions. Hence linear 
motion polynomials encode this kind of motions. It follows that, if we are 
able to factor a motion polynomial into linear ones, then we can decompose a rational 
motion into revolutions and translational motions. Later we will consider a 
suitable subclass of motion polynomials, so that only the first situation 
occurs.

\begin{remark}
\label{remark:multiply_real}
Let $\phi\colon \R \dashrightarrow \p^3_{\R}$ be a rational motion and
let $P(t) \in \K[t]$ be the corresponding motion polynomial, so $P = Z + \eta
\tth W$ with $Z, W \in \C[t]$. Notice that, although $Z$ and $W$ have 
coefficients in~$\C$, the polynomial~$P$ determines a real curve 
in~$\p^3_{\R}$, namely a curve described by real equations. Hence, if $R \in 
\R[t]$ is a nonzero real polynomial, one can check that $R P \in 
\K[t]$ provides the same curve in~$\p^3_{\R}$ as~$P$.
\end{remark}

\begin{remark}
\label{remark:multiply_complex}
Let again $\phi\colon \R \dashrightarrow \p^3_{\R}$ be a rational motion and
let $P(t) \in \K[t]$ be the corresponding motion polynomial. For any $C\in\C[t]$
the orbits of the origin under the motions given by $P$ and $C P$ are the 
same since $C$ represents a revolution around the origin --- recall that $CP$ 
acts on $\R^2$ by first applying $C$ and then~$P$.
\end{remark}

\begin{example}
The rational motion given by the polynomial $t + \eta$ is a vertical
translational motion. Indeed, its image in $\p^3_{\R}$ is parametrized by
$(\lambda:0:\mu:0)$, so that we are in case~\ref{lemma:one_dim_motions:case1}
of Lemma~\ref{lemma:one_dim_motions}. In contrast, the motion polynomial $P(t) =
t + \imath$ gives a revolution around the origin $(0,0)$. This is an instance
of case~\ref{lemma:one_dim_motions:case2} of Lemma~\ref{lemma:one_dim_motions},
using the parametrization $(\lambda:\mu:0:0)$.
\end{example}
\begin{example}
\label{example:circular_motion}
Consider the following product of two linear motion polynomials:
  \[ \underbrace{(t + \imath)}_{\substack{\mathrm{rotation\ around} \\
\mathrm{a\ point}}} \cdot \underbrace{(t - \imath +
\eta)}_{\substack{\mathrm{rotation\ in\ the\ opposite} \\ \mathrm{sense,\
around\ another\ point}}} \; = \; \underbrace{(t^2 + 1) + \eta \tth (t -
\imath).}_{\substack{\mathrm{this\ is\ a\ translation,\ since} \\ \mathrm{there\
is\ no\ imaginary\ part\ in\ } t^2 + 1}}\]
For a fixed $t$ we get a translation by the vector $\frac{1}{t^2 +
  1}(t,-1)$. As $t$ changes, this vector describes a circle, hence we get a 
translational motion along a circle (see
Figure~\ref{figure:circular_translation}).

\begin{figure}[ht]
	\resizebox{5cm}{!}{
	\includegraphics{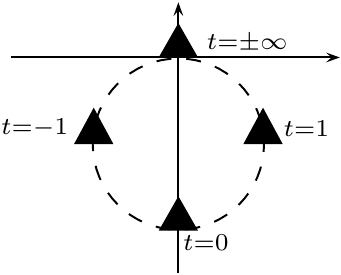}
	}
	\caption{The trace of the point $(0,0)$ under the circular translation given 
by the motion polynomial $(t^2 + 1) + \eta \tth (t - \imath)$ is depicted via a 
dashed line. The action of the motion on the black triangle highlights the fact 
that we have a purely translational motion.}
	\label{figure:circular_translation}
\end{figure}
\end{example}

In Example~\ref{example:circular_motion} we saw that the factorization of a 
motion polynomial into linear polynomials provides a decomposition of the 
described motion into revolutions. However, 
Example~\ref{example:elliptic_motion} shows that this is not always possible.

\begin{example}
\label{example:elliptic_motion}
Let us consider an \emph{elliptic translational motion}:
\[
  P(t) \; = \; (t^2 + 1) + \eta \tth (at - b\imath), \qquad a,b \in \R.
\]
We try to factor $P$ into two linear polynomials $P_1, P_2 \in \K[t]$. First
notice that the primal parts of $P_1$ and $P_2$ should be factors of $(t^2 +
1)$, so they have to be of the form $t \pm \imath$. Thus we have only two
possibilities:
	\[ \left\{ \begin{array}{rcl}
		\ppt(P_1) & = & t + \imath \\
		\ppt(P_2) & = & t - \imath
	\end{array} \right. \qquad \qquad
	\left\{ \begin{array}{rcl}
		\ppt(P_1) & = & t - \imath \\
		\ppt(P_2) & = & t + \imath
	\end{array} \right. \]
	By a direct computation one can prove that none of those two choices gives a
factorization if $a \neq b$, thus $P$ cannot be factored into linear
polynomials. 

	On the other hand, we show now that if we multiply~$P$ by a real 
polynomial~$R$ we can achieve a factorization. Recall that, by 
Remark~\ref{remark:multiply_real}, the polynomials~$P$ and~$RP$ describe 
the same motion. In this case we take $R = t^2 + 1$. Hence we need four linear
polynomials $P_1, \dotsc, P_4$ to factorize $RP$, and again $\ppt(P_i) =
t \pm \imath$. We make the following ansatz:
	\begin{equation}
        \label{equation:ansatz_elliptic}
		\begin{array}{lcl}
			\begin{array}{rcl}
				P_1 & = & (t - \imath) + \eta \tth w_1, \\
				P_2 & = & (t + \imath) + \eta \tth w_2, \\
			\end{array} & & 
			\begin{array}{rcl}
				P_3 & = & (t + \imath) + \eta \tth w_3, \\
				P_4 & = & (t - \imath) + \eta \tth w_4. \\
			\end{array}
		\end{array}
	\end{equation}
	Imposing $P_1 \cdots P_4 = RP$ gives a linear system in the $w_i$,
namely:
\begin{equation}
\label{equation:solution_elliptic}
	\left\{ \begin{array}{rcl}
		w_1 + w_2 + w_3 + w_4 & = & a, \\
		w_1 + w_2 - w_3 - w_4 & = & b.
	\end{array} \right.
\end{equation}
	Each solution of the system~\eqref{equation:solution_elliptic} gives rise to 
a factorization of the polynomial~$RP$. If, instead, we take the following 
ansatz:
\[
	\begin{array}{lcl}
		\begin{array}{rcl}
			P_1 & = & (t + \imath) + \eta \tth w_1, \\
			P_2 & = & (t - \imath) + \eta \tth w_2, \\
		\end{array} & & 
		\begin{array}{rcl}
			P_3 & = & (t + \imath) + \eta \tth w_3, \\
			P_4 & = & (t - \imath) + \eta \tth w_4. \\
		\end{array}
	\end{array}
\]
	one can check that the corresponding linear system does not admit any
solution. 
\end{example}

\noindent In the next section, we consider the factorization problem in a more
systematic way. We will see that, once we restrict to a certain class of motion 
polynomials, all situations can be treated as we did in 
Example~\ref{example:elliptic_motion}: in general a motion polynomial~$P$ 
cannot be factored into linear polynomials in~$\K[t]$, but it is possible to 
find a real polynomial~$R$ such that $R P$ can be factored.

\smallskip
We end this section by defining precisely what we mean by saying that a
linkage \emph{realizes} a given rational motion. Recall from 
Definition~\ref{definition:relative_position} that, given two links~$i$ and~$j$ 
of a linkage~$L$, we can define the set~$\rpos{i}{j}$ of relative positions of 
the link~$j$ with respect to the link~$i$, and this is a subset of~$\SE{2}$.
Furthermore, recall that the set~$\conf{L}$ of configurations of~$L$ has the
structure of a real projective variety. We are going to use some concepts and 
results coming from real algebraic geometry: 
\begin{enumerate}
  \item\label{ag1}
  every projective space~$\p^n_{\R}$ is a real affine variety, namely 
there exists an embedding $\xi \colon \p^n_{\R} \hookrightarrow \R^k$ for some 
$k \in \N$, and $\xi\bigl(\p^n_{\R}\bigr)$ is the zero set of a collection of 
real polynomials; see for example~\cite[Theorem~3.4.4]{Bochnak1998};
  \item\label{ag2}
  a subset $X \subseteq \R^k$ is called \emph{semialgebraic} if it can be 
described as the set of points satisfying a disjunction of conjunctions of 
polynomial equalities and inequalities; see for 
example~\cite[Definition~2.1.4]{Bochnak1998};
  \item\label{ag3}
  a subset $X \subseteq \p^n_{\R}$ is called semialgebraic if its image under 
any embedding $\xi \colon \p^n_{\R} \hookrightarrow \R^k$ as in~\ref{ag1} is a 
semialgebraic set in the sense of~\ref{ag2}.
\end{enumerate}

\begin{lemma}
	Let $L$ be a linkage and let $i,j$ be links of~$L$. Then, under the 
projective embedding of~$\SE{2}$, the set~$\rpos{i}{j}$ becomes a semialgebraic 
subset of~$\SE{2}$.
\end{lemma}
\begin{proof}
Let $(i, h_1), \dotsc, (h_s, j)$ be any directed path from $i$ to~$j$. Consider 
the map $F \colon \conf{L} \longrightarrow \SE{2}$ sending a 
configuration $\Sigma = (\sigma_{k,l})$ to $\sigma_{i,h_1} \circ 
\dotsb \circ \sigma_{h_s, j}$. One can check that~$\rpos{i}{j}$ coincides with 
$F(\conf{L})$. Moreover, if we write the map~$F$ in terms of the 
projective coordinates of $\conf{L}$ and $\SE{2}$, then $F$ is given by real 
polynomials. Hence, because of a general result in real algebraic geometry (see 
for example~\cite[Theorem~1.4.2]{Bochnak1998}), the set~$\rpos{i}{j}$ is 
semialgebraic.
\end{proof}

In the following definition, and later in this paper, the word ``component'' 
will always stand for ``irreducible component'' with respect to the Zariski 
topology.

\begin{definition}
\label{definition:realizing}
Let $L$ be a linkage and let $\phi \colon \R \dashrightarrow \p^3_{\R}$ be a 
rational motion. We say that
\begin{enumerate}
\item\label{definition:realizing:case1}
  $L$ \emph{weakly realizes} the motion~$\phi$ if there exist links~$i$
  and~$j$ of~$L$ such that $\mathrm{image}(\phi) \subseteq \rpos{i}{j}$.
\item\label{definition:realizing:case2}
  $L$ \emph{strongly realizes} the motion~$\phi$ if there exist links~$i$ 
  and~$j$ of $L$ such that $\mathrm{image}(\phi) \subseteq \rpos{i}{j}$ and 
  $\overline{\mathrm{image}(\phi)}$ is a component of~$\rpos{i}{j}$, where
  $\overline{(\cdot)}$ denotes the Zariski closure in~$\rpos{i}{j}$.
\end{enumerate}
\end{definition}
\begin{figure}[ht]
\begin{center}
\includegraphics[width=0.4\textwidth]{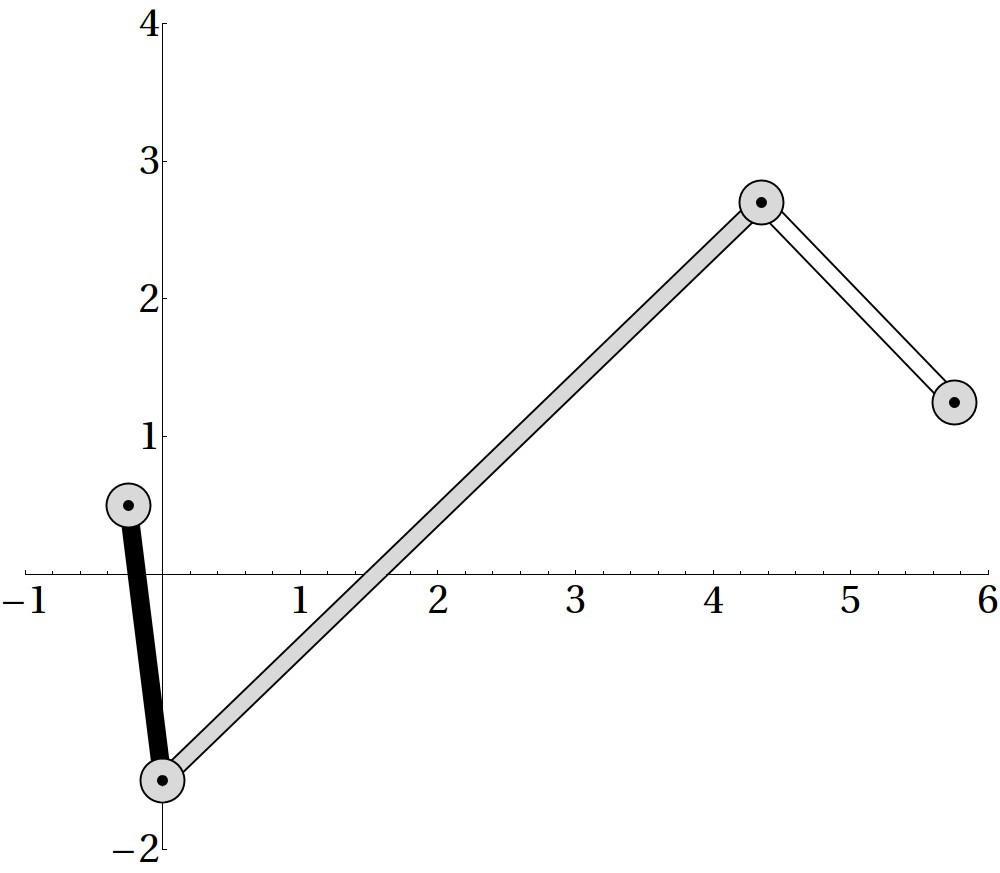}
\qquad
\includegraphics[width=0.4\textwidth]{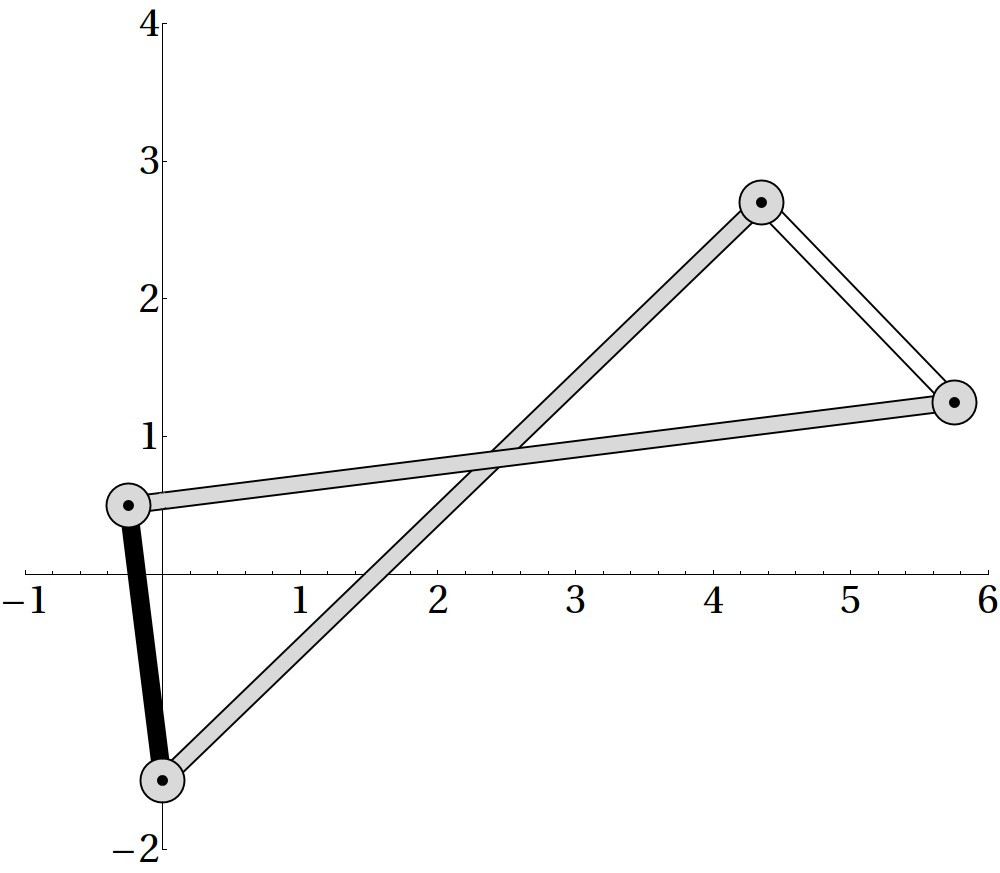}
\end{center}
\caption{Two linkages that both realize the motion given by
$\bigl(t+\imath-3\eta\bigr)\cdot\bigl(t+2\imath+\eta(18+3\imath)\bigr)$, discussed
in Example~\ref{example:mobility}. The depicted position is reached at~$t=-2$.  
The linkage on the left has two degrees of freedom and thus weakly realizes the 
given motion as the relative position between base (in black) and end effector 
(in white). The linkage on the right has mobility one and provides a strong 
realization of this motion.}
\label{figure:weak_strong}
\end{figure}
\begin{remark}
Notice that in case~\ref{definition:realizing:case2}
of Definition~\ref{definition:realizing} we allow the
given variety of relative positions to have several components, and we ask that
only one of them coincides with the Zariski closure of the image of the given
rational motion, which is irreducible by construction. This situation occurs in
the example given on the right of Figure~\ref{figure:weak_strong},
since the variety of relative positions between base and end effector has two
components (see Lemma~\ref{lemma:flip_mobility}).
\end{remark}

\section{Factorization of motion polynomials}
\label{factorization}

Motivated by the considerations from the previous section, we now address the
problem of factorizing a motion polynomial $P = Z + \eta \tth W \in \K[t]$
into linear factors.  The main results are
Theorem~\ref{theorem:bounded_factor} and the corresponding algorithm
\texttt{FactorMotionPolynomial}. To achieve them, we have to go through some
technicalities, which however will not be needed in the rest of the paper.

We restrict our attention to monic polynomials, i.e., to polynomials~$P$ whose
leading coefficient is~$1$, which implies that $\deg{W} < \deg{Z}$.  Our goal
is to write a monic $P \in \K[t]$ in the form $P = P_1 \cdots P_n$, where $n =
\deg{P}$ and $\deg{P_i} = 1$ for all~$i$.  If $P = Z + \eta \tth W$ and
$P_i=Z_i+\eta\,W_i$, then $Z = Z_1 \cdots Z_n$. Hence the primal part of
each~$P_i$ has to be one of the linear factors of $Z\in\C[t]$.  We start by a
few definitions and results characterizing monic motion polynomials that split
into linear factors.

\begin{definition}
\label{definition:Q_i}
Let $\mathbf{z} = (z_1, \dotsc, z_n)\in\C^n$. We define the following polynomials:
\[ Q_i(\mathbf{z}) \; := \; (t - \overline{z_1}) \cdots (t - \overline{z_{i-1}})
   (t - z_{i+1}) \cdots (t - z_{n}) \quad \mathrm{for\ } i \in \{1, \dotsc, n \}.
\]
Hence we have that $\deg{Q_i(\mathbf{z})} = n - 1$ for all~$i$.
\end{definition}
\begin{remark}
\label{remark:gcd_Q_i}
Let $Z\in\C[t]$ and let $\mathbf{z} = (z_1, \dotsc, z_n)$ be a fixed
permutation of the roots of~$Z$. Then the polynomials $Q_i(\mathbf{z})$ have a
non-trivial $\gcd$ if and only if the polynomial~$Z$ admits a pair of
complex-conjugate roots $\alpha$ and $\overline{\alpha}$ (this includes the
case of multiple real roots).
\end{remark}

\begin{lemma}
\label{lemma:characterization_factorization}
Let $P=Z+\eta\,W \in \K[t]$ be a monic motion polynomial, and let $\mathbf{z}
= (z_1, \dotsc, z_n)$ be a fixed permutation of the roots of $Z$ over~$\C$.
Then $P$ admits a factorization $P = P_1 \cdots P_n$, where $P_i(t) = (t -
z_i) + \eta \tth w_i$ with $w_i \in \C$, if and only if $W$ lies in
$\bigl\langle Q_1(\mathbf{z}), \dotsc, Q_n(\mathbf{z}) \bigr\rangle_{\C}$,
i.e., in the $\C$-linear span of $Q_1(\mathbf{z}), \dotsc,
Q_n(\mathbf{z})$.
\end{lemma}
\begin{proof}
	Suppose that $P \in \K[t]$ admits a factorization into linear factors:
	\[
		P(t) = \prod_{i=1}^n (t-z_i+\eta w_i) = \prod_{i=1}^n(t-z_i) +
\eta\sum_{k=1}^n\underbrace{\bigg(\prod_{j=1}^{k-1}(t-\overline{z_j}
)\bigg)\bigg(\prod_{j=k+1}^n(t-z_j)\bigg)}_{\textstyle =Q_k(\mathbf{z})} w_k.
	\]
	Such a factorization exists if and only if we can choose
$w_1, \dotsc, w_n\in\C$ such that
	$\sum_{k=1}^n w_k Q_k(\mathbf{z})$ matches the prescribed $\eta$-part of 
$P(t)=Z(t)+\eta W(t)$, namely if and only if $W \in \bigl\langle 
Q_1(\mathbf{z}), \dotsc, Q_n(\mathbf{z}) \bigr\rangle_{\C}$.
\end{proof}

\begin{lemma}
\label{lemma:no_conj_roots}
Let $P \in \K[t]$ be a monic motion polynomial and suppose that $Z = \ppt(P)$ has
no pair of complex-conjugate roots.
Then for every permutation $\mathbf{z} = (z_1, \dotsc, z_n)$ of the roots of~$Z$,
the polynomial~$P$ admits a factorization $P = P_1 \cdots P_n$ with
$\ppt(P_i) = t - z_i$ for all $i$.
\end{lemma}
\begin{proof}
We fix a permutation~$\mathbf{z}$ of the roots of~$Z$, and we 
denote $Q_k = Q_k(\mathbf{z})$ for all~$k$. From 
Lemma~\ref{lemma:characterization_factorization}, we know that
$P = Z + \eta \tth W$ admits a factorization if and only if $W \in \bigl\langle
Q_1, \dotsc, Q_n \bigr\rangle_{\C}$. Clearly, this is always possible (for
arbitrary~$W$) if the following matrix $M_n\in\C^{n\times n}$ is non-singular:
\[
  M_n = \begin{pmatrix}
   \langle t^0\rangle Q_1 & \cdots & \langle t^0\rangle Q_n \\
   \langle t^1\rangle Q_1 & \cdots & \langle t^1\rangle Q_n \\
   \vdots                 &       & \vdots \\
   \langle t^{n-1}\rangle Q_1 & \cdots & \langle t^{n-1}\rangle Q_n \\
  \end{pmatrix}.
\]
Here $\langle t^i\rangle Q_k$ denotes the coefficient of $t^i$ in the
polynomial~$Q_k$. Notice that the matrix entries are, up to sign, elementary
symmetric polynomials in the $z_i$ and $\overline{z_i}$. We now exhibit that
the determinant of~$M_n$ is nonzero under the assumptions on the roots
of~$Z$. Indeed, we find that
\begin{equation}\label{eq:det}
  \det(M_n) = \prod_{1\leq i<j\leq n}\bigl(\overline{z_i}-z_j\bigr),
\end{equation}
which follows, as a special case, from Lemma~3 in~\cite{Krattenthaler99}
(by substituting \mbox{$A_k=-z_k$}, $B_k=-\overline{z_{k-1}}$, and by extracting
the coefficient of $\prod_{i=1}^n X_i^{i-1}$).  A similar determinant
evaluation is given in~\cite{LascouxPragacz02} where the $z_i$ appear without
conjugation. Since $\det(M_n)$ is very much reminiscent of the Vandermonde
determinant, it is no surprise that~\eqref{eq:det} can also be proved in an
analogous fashion.
\end{proof}

\begin{cor}
\label{cor:span}
  Let $\z=(z_1,\dotsc,z_n)\in\C^n$ such that 
$\gcd\bigl(Q_1(\z),\dotsc,Q_n(\z)\bigr) = 1$. Then 
\[ 
  \bigl\langle Q_1(\z),\dotsc,Q_n(\z) \bigr\rangle_{\C} \; = \; \bigl\{ W \in 
\C[t] \, : \, \deg(W) < n \bigr\}.
\]
\end{cor}
\begin{proof} 
  The gcd condition implies that $\z$ has no pair of complex-conjugate roots, as 
noticed in Remark~\ref{remark:gcd_Q_i}. Then the claim follows from 
Lemma~\ref{lemma:characterization_factorization} and 
Lemma~\ref{lemma:no_conj_roots}.
\end{proof}

Since the $\gcd$ of the polynomials $Q_i(\mathbf{z})$ introduced in
Definition~\ref{definition:Q_i} plays an important role for the factorization
of motion polynomials, we make some effort to describe it precisely
(Lemma~\ref{lemma:interlacement}, Corollary~\ref{cor:interlacement},
Proposition~\ref{proposition:gcd}); for this purpose the combinatorial notion
of \emph{matching} will be introduced. At the same time we aim at a nicer
characterization of those secondary parts in
Lemma~\ref{lemma:characterization_factorization} that ensure the existence of a
factorization (Proposition~\ref{proposition:factorization_ideal}).
To simplify the proofs of Propositions~\ref{proposition:gcd}
and~\ref{proposition:factorization_ideal}, we formulate the somewhat technical
Lemma~\ref{lemma:technical}.

\begin{lemma}
\label{lemma:interlacement}
Let $\alpha,\beta\in\C$ with $\alpha\neq\beta\neq\albar$, and let
$\z=(z_1,\dotsc,z_n)\in\C^n$ and
$\zt=(z_1,\dotsc,z_{j-1},\beta,z_j,\dotsc,z_n)\in\C^{n+1}$ for some
$j\in\{1,\dotsc,n+1\}$.  Then the multiplicities of the roots $\alpha$ and
$\albar$ in $\gcd\bigl(Q_1(\z),\dotsc,Q_n(\z)\bigr)$ and in
$\gcd\bigl(Q_1(\zt),\dotsc,Q_{n+1}(\zt)\bigr)$ are the same.
\end{lemma}
\begin{proof}
First observe that $Q_k(\z)$ and $Q_\ell(\zt)$ have the same multiplicities 
of~$\alpha$ and~$\albar$, where $\ell=k$ if $k<j$ and $\ell=k+1$ if $k\geq j$.
On the other hand, the polynomial~$Q_j(\zt)$ has no such counterpart among the
$Q_k(\z)$, but its multiplicities of~$\alpha$ and~$\albar$ are greater than or
equal to those in its (one or two) neighbors $Q_{j-1}(\zt)$ and
$Q_{j+1}(\zt)$. The claim follows.
\end{proof}

\begin{cor}
\label{cor:interlacement}
Let $\mathbf{x}\in\C^m$ and $\mathbf{y}\in\C^n$ with $x_i\neq y_j\neq
\overline{x_i}$ for all $1\leq i\leq m$ and $1\leq j\leq n$. If
$\z\in\C^{m+n}$ is an arbitrary interlacement of $\mathbf{x}$ and
$\mathbf{y}$, i.e., there are $1\leq i_1\leq\dotsb\leq i_m\leq m+n$ and $1\leq
j_1\leq\dotsb\leq j_n\leq m+n$ with
$\{i_1,\dotsc,i_m\}\cap\{j_1,\dotsc,j_n\}=\emptyset$ such that $z_{i_k}=x_k$,
$1\leq k\leq m$, and $z_{j_k}=y_k$, $1\leq k\leq n$, then
\[
  \gcd\bigl(Q_1(\z),\dotsc,Q_{m+n}(\z)\bigr) =
  \gcd\bigl(Q_1(\mathbf{x}),\dotsc,Q_m(\mathbf{x})\bigr) \cdot
  \gcd\bigl(Q_1(\mathbf{y}),\dotsc,Q_n(\mathbf{y})\bigr).
\]
\end{cor}
\begin{proof}
This is a direct consequence of Lemma~\ref{lemma:interlacement}.
\end{proof}

\begin{definition}
\label{definition:matching}
Let $\z=(z_1, \dotsc, z_n)\in\C^n$. A set
\[
  M \; \subseteq \; \bigl\{(i,j): \, 1 \leq i < j \leq n \; \mathrm{\ and\ } \;
z_i = \overline{z_j} \bigr\}
\]
is called a \emph{matching} of~$\z$ if for all $(i_1,j_1),(i_2,j_2) \in M$ we
have $i_1\neq i_2$ and $j_1\neq j_2$. (Note that for each $\z\in\C^n$ there is
a bipartite graph with at most $2n-2$ vertices such that the matchings in that
graph and the matchings of~$\z$ are in bijection.)
\end{definition}

\begin{lemma}
\label{lemma:technical}
Let $\z=(z_1,\dotsc,z_n)\in\C^n$, let $G=\gcd\bigl(Q_1(\z),\dotsc,Q_n(\z)\bigr)$,
and let $g=\deg G$. For a matching~$M$ of $\z$ define $H(M):=\prod_{(i,j)\in M}(t-z_j)$.
Then the following statements hold:
\begin{enumerate}
\item\label{ass:1} For any matching $M$ of $\z$ we have $H(M)\mid G$.
\item\label{ass:2} There is a matching~$M$ of~$\z$ such that $H(M)=G$.
\item\label{ass:3} There are indices $1\leq i_1\leq\dotsb\leq i_{n-g}\leq n$ 
such that, setting $\zt = \bigl(z_{i_1},\dotsc, 
z_{i_{n-g}}\bigr)$, we have $Q_k(\zt)=Q_{i_k}(\z)/G$ for all $1\leq k\leq n-g$ 
and $\gcd\bigl( Q_1(\zt), \dotsc, Q_{n-g}(\zt)\bigr) = 1$.
\end{enumerate}
\end{lemma}
\begin{proof}
For \ref{ass:1} we show that $H(M)$ divides~$Q_k(\mathbf{z})$ for any 
matching~$M$ of~$\z$ and for any $k\in\{1,\dotsc,n\}$. Indeed, by rewriting
\[
  H(M) = \prod_{(i,j)\in M}(t-z_j) =
  \bigg(\prod_{\genfrac{}{}{0pt}{1}{(i,j)\in M}{j>k}}(t-z_j)\bigg)
  \bigg(\prod_{\genfrac{}{}{0pt}{1}{(i,j)\in M}{j\leq
k}}(t-\overline{z_i})\bigg)
\]
the claim follows, since $Q_k(\mathbf{z})$ has precisely the roots
$\overline{z_1}, \dotsc, \overline{z_{k-1}}$, $z_{k+1}, \dotsc, z_n$.

For \ref{ass:2} and \ref{ass:3} note that the statements trivially hold
for $G=1$, so we assume now that $G\neq1$.  Using
Corollary~\ref{cor:interlacement} the proof is reduced to the case
where $\z\in\{\alpha,\albar\}^n$ for some $\alpha\in\C$. If $\alpha\in\R$ then
we have $Q_1(\z)=\dotsb=Q_n(\z)=(t-\alpha)^{n-1}=G$ and we can choose
$M:=\{(1,2),(2,3),\dotsc,(n-1,n)\}$ and $i_k=k$ for $1\leq k<n$. Hence from now
on we assume $\alpha\neq\albar$.

Both assertions \ref{ass:2} and \ref{ass:3} are proved in parallel by an
inductive argument on the degree of~$G$, so in the following we show how to
reduce the given scenario to one with smaller~$g$. The main difficulty here is
to pick an index~$j$ such that deleting~$z_j$ from~$\z$ yields a $\gcd$ of
smaller degree. We distinguish three cases and assume w.l.o.g. that
$z_n=\alpha$ (case A) resp. $z_1=\alpha$ (cases B and C).

\smallskip
\emph{Case A:} There exists an index~$j$ such that
$\#_{\alpha}(z_j,\dotsc,z_n)\leq\#_{\albar}(z_j,\dotsc,z_n)$, where
$\#_{\alpha}$ denotes the number of occurrences of $\alpha$. In this situation
we set $\zt=(z_1,\dotsc,z_{n-1})$ and
$\tilde{G}=\gcd\bigl(Q_1(\zt),\dotsc,Q_{n-1}(\zt)\bigr)$. By the induction
hypothesis there exists a matching $\tilde{M}$ for $\zt$ such that
$H(\tilde{M})=\tilde{G}$. Because of the assumption on the number of
$\alpha$'s and $\albar$'s, there must be an index $i\in\{j,\dotsc,n-1\}$ with
$z_i=\albar$ that does not occur as the first entry of any pair
in~$\tilde{M}$. Hence $M:=\tilde{M}\cup\{(i,n)\}$ is a matching for~$\z$ and
we have $H(M)=(t-\alpha)H(\tilde{M})$. Since by \ref{ass:1} we have
$H(M)\mid G$, we obtain $(t-\alpha)\tilde{G}\mid G$. Next observe that
$Q_k(\zt)=Q_k(\z)/(t-z_n)$ for all $1\leq k<n$ and hence
\[
  \tilde{G} = \gcd\bigl(Q_1(\zt),\dotsc,Q_{n-1}(\zt)\bigr) =
  \frac1{t-\alpha} \gcd\bigl(Q_1(\z),\dotsc,Q_{n-1}(\z)\bigr).
\]
From the definition of~$G$ it follows that $G\mid(t-\alpha)\tilde{G}$,
and thus $G=(t-\alpha)\tilde{G}=H(M)$.

\smallskip
\emph{Case B:} There exists an index~$j$ such
$\#_{\alpha}(z_1,\dotsc,z_j)\leq\#_{\albar}(z_1,\dotsc,z_j)$. The reasoning
in this case is completely analogous to case A, setting $\zt=(z_2,\dotsc,z_n)$.

\smallskip
\emph{Case C:} Neither case A nor case B applies, which means that for all
indices~$j$ we have $\#_{\alpha}(z_1,\dotsc,z_j)>\#_{\albar}(z_1,\dotsc,z_j)$
and $\#_{\alpha}(z_j,\dotsc,z_n)>\#_{\albar}(z_j,\dotsc,z_n)$. Note that in this
situation we automatically have $z_1=z_n$, which w.l.o.g. we assumed to
be~$\alpha$.  Now there exists no entry in~$\z$ whose removal would yield a
$\gcd$ of degree $\deg G-1$. We circumvent this problem by setting
$\zt=(z_1,\dotsc,z_{j-1},z_{j+2},\dotsc,z_n)$ where $j$ is the largest index
such that $z_j=\albar$; in particular $j<n$. The rest is analogous to the
previous reasoning. As before, we may assume that there exists a
matching~$\tilde{M}$ of $\zt$ with $H(\tilde{M})=\tilde{G}=
\gcd\bigl(Q_1(\zt),\dotsc,Q_{n-2}(\zt)\bigr)$. Since $\#_{\alpha}(z_1,\dotsc,z_j)
> \#_{\albar}(z_1,\dotsc,z_j)$ there exists an index $i<j$ with $z_i=\alpha$
such that $i$ does not appear as the first entry of any pair in~$\tilde{M}$.
In order to extend $\tilde{M}$ to a matching for~$\z$ we need some relabeling;
for this purpose we define $\tau(a)=a$ if $a<j$ and $\tau(a)=a+2$ otherwise.
Then 
\[ 
  M:=\bigl\{\bigl(\tau(a),\tau(b)\bigr)\mid(a,b)\in\tilde{M}\bigr\}\cup 
     \bigl\{(i,j),(j,j+1)\bigr\}
\]
is a matching for~$\z$ and $H(M)=(t-\alpha)(t-\albar)H(\tilde{M})$;
hence $(t-\alpha)(t-\albar)\tilde{G}\mid G$. A simple calculation shows that
$Q_k(\zt)=Q_k(\z)/\bigl((t-\alpha)(t-\albar)\bigr)$ for $k<j$ and that
$Q_k(\zt)=Q_{k+2}(\z)/\bigl((t-\alpha)(t-\albar)\bigr)$ for $k\geq j$. Thus
we get
\[
  \tilde{G} = \frac1{(t-\alpha)(t-\albar)}
  \gcd\bigl(Q_1(\z),\dotsc,Q_{j-1}(\z),Q_{j+2}(\z),\dotsc,Q_n(\z)\bigr).
\]
Again, this implies $G=(t-\alpha)(t-\albar)\tilde{G}=H(M)$,
which concludes the proof.
\end{proof}

\begin{proposition}
\label{proposition:gcd}
Let $\z=(z_1,\dotsc,z_n)\in\C^n$ and let $M$ be a matching of~$\mathbf{z}$ of
maximal size, and let $Q_1(\z),\dotsc,Q_n(\z)$ be as in
Definition~\ref{definition:Q_i}. Then we have
\[
  G := \gcd\bigl(Q_1(\z), \dotsc, Q_n(\z)\bigr) = \prod_{(i,j)\in M}(t-z_j)
\]
(where the $\gcd$ is assumed to be a monic polynomial).
\end{proposition}
\begin{proof}
In Lemma~\ref{lemma:technical} part~\ref{ass:1} it was shown that the
product on the right-hand side, which we denoted by $H(M)$, divides~$G$ for
any matching~$M$. For the other direction we have to argue that the maximality
of~$M$ implies that $G\mid H(M)$, or equivalently, that $G\nmid H(M)$ implies
that $M$ is not maximal. From Lemma~\ref{lemma:technical} part~\ref{ass:2} it
follows that there exists a matching $M'$ of $\z$ such that $H(M')=G$.  Note
that $|M'|=\deg(G)$. Thus, if $G\nmid H(M)$, which means that $\deg(H(M))<\deg(G)$,
then $M$ is not maximal because $|M|<|M'|$.
\end{proof}

\begin{example}
Let $Z=(t-\alpha)^r(t-\overline{\alpha})^{r+1}$ for some
$\alpha\in\C\setminus\R$. In the following table we consider different
permutations $\z\in\{\alpha,\albar\}^{2r+1}$ of the roots of~$Z$; to each
permutation we give the gcd $G = \gcd\bigl(Q_1(\z), \dotsc, Q_n(\z)\bigr)$ and a 
matching~$M$ of maximal size that witnesses~$G$:
\medskip
\begin{center}
\begin{tabular}{c c c}
\hline permutation & \rule{0pt}{10pt} $G$ & $M$ \\ \hline 
\rule{0pt}{12pt} 
$(\alpha,\dotsc,\alpha,\overline{\alpha},\dotsc,\overline{\alpha })$ &
$(t-\overline{\alpha})^r$ &
  $\{(1,r+1),(2,r+2),\dotsc,(r,2r)\}$ \\
$(\overline{\alpha},\dotsc,\overline{\alpha},\alpha,\dotsc,\alpha)$ &
$(t-\alpha)^r$ &
  $\{(1,r+2),(2,r+3),\dotsc,(r,2r+1)\}$ \\
$(\overline{\alpha},\alpha,\overline{\alpha},\alpha,\dotsc,\alpha,\overline{
\alpha})$ &
  $(t-\alpha)^r(t-\overline{\alpha})^r$ & $\{(1,2),(2,3),\dotsc,(2r,2r+1)\}$ \\
%\hline
\end{tabular}
\end{center}
\medskip
The cases displayed above are the extreme ones: indeed, it is easy to see that
$r\leq\deg(G)\leq 2r$. Moreover, for any $G=(t-\alpha)^i(t-\overline{\alpha})^j$
with $0\leq i,j\leq r$ and $i+j\geq r$ there exists a permutation~$\mathbf{z}$
that produces exactly this gcd~$G$. These considerations lead to interesting
combinatorial questions --- e.g., how many permutations of a given set of roots
are there that produce a prescribed gcd --- which, however, are irrelevant for
our purposes.
\end{example}

\begin{proposition}
\label{proposition:factorization_ideal}
Let $\mathbf{z} = (z_1, \dotsc, z_n)\in\C^n$ and $W\in\C[t]$ with $\deg W<n$. Then
\[
  W \in \bigl\langle Q_1(\mathbf{z}), \dotsc, Q_n(\mathbf{z}) \bigr\rangle_{\C} 
\quad \iff \quad W \in \bigl( Q_1(\mathbf{z}), \dotsc, Q_n(\mathbf{z}) \bigr) 
\cdot \C[t],
\]
where $\bigl( Q_1(\mathbf{z}), \dotsc, Q_n(\mathbf{z}) \bigr) \cdot \C[t]$ is the 
ideal of $\C[t]$ generated by $Q_1(\mathbf{z}), \dotsc, Q_n(\mathbf{z})$.
\end{proposition}
\begin{proof}
The claim is equivalent to the following statement:
\begin{equation} 
\label{equation:restatement_ideal}
  \bigl\langle Q_1(\z), \dotsc, Q_n(\z) \bigr\rangle_{\C} \; = \; \bigl( Q_1(\z),
  \dotsc, Q_n(\z) \bigr) \cdot \C[t] \, \cap \, \C[t]_{<n},
\end{equation}
where $\C[t]_{<n}$ denotes the set of complex polynomials of degree less
than~$n$. Note that the containment ``$\subseteq$'' is trivial. Let $G = 
\gcd{\left( Q_1(\z), \dotsc, Q_n(\z) \right)}$. If \mbox{$G=1$}, then by 
Corollary~\ref{cor:span} we have $\bigl\langle Q_1(\z), \dotsc, Q_n(\z) 
\bigr\rangle_{\C} = \C[t]_{<n}$, which implies 
Equation~\eqref{equation:restatement_ideal}. On the other hand, if $G$ is 
non-trivial, Equation~(\ref{equation:restatement_ideal}) is equivalent~to
\begin{equation}
\label{equation:second_restatement_ideal}
  \bigl\langle \tilde{Q}_1, \dotsc, \tilde{Q}_n \bigr\rangle_{\C} \; = \;
  \bigl( \tilde{Q}_1, \dotsc, \tilde{Q}_n \bigr) \cdot \C[t] \, \cap \, 
\C[t]_{<n-g},
\end{equation}
where $\tilde{Q}_i = Q_i(\z)/G$ and $g = \deg{G}$. By 
Lemma~\ref{lemma:technical} part~\ref{ass:3}, there exists 
$\zt=(z_{i_1},\dotsc,z_{i_{n-g}})$ 
such that $\gcd\bigl(Q_1(\zt), \dotsc, Q_{n-g}(\zt)\bigr) = 1$, and $Q_k(\zt) 
= \tilde{Q}_{i_k}$ for all~$k$. Hence, by Corollary~\ref{cor:span} we have 
$\bigl\langle Q_1(\zt), \dotsc, Q_{n-g}(\zt) \bigr\rangle_{\C} = 
\C[t]_{<n-g}$. This implies 
Equation~\eqref{equation:second_restatement_ideal}, since $\bigl\langle 
Q_1(\zt), \dotsc, Q_{n-g}(\zt) \bigr\rangle_{\C} \; \subseteq \;
\bigl\langle \tilde{Q}_1, \dotsc, \tilde{Q}_n \bigr\rangle_{\C}$.
\end{proof}

Recall that our goal is to decompose a given motion into a sequence of
revolutions, by factorizing a corresponding motion polynomial. For this
purpose we have to restrict the domain of motion polynomials we are working
with: for example, there is no hope to write a motion whose orbits are
unbounded as the composition of revolutions, since the orbits of revolutions
are always bounded.  In particular, if we consider a rational bounded curve,
then it can be given by a parametrization $\varphi=(f/h,g/h)$ for some real
polynomials $f,g$ and~$h$ with $\deg{h} \ge \max \{ \deg{f}, \deg{g} \}$. The
boundedness of the curve implies that $h$ has no real roots. The limit 
of~$\varphi(t)$ for $t \rightarrow \infty$ gives a point in~$\R^2$, which can be 
translated to the origin. Then we even have $\deg{h} > \max \{ \deg{f}, \deg{g} 
\}$.  If we further assume that $h$ is monic, then also the motion polynomial $P 
= h + \eta \tth (f + \imath g)$ is monic, and by 
Proposition~\ref{prop:curve_motion} it traces the image of~$\varphi$. This 
motivates the following definition.

\begin{definition}
\label{definition:bounded_poly}
Let $P = Z + \eta \tth W$ be a motion polynomial in~$\K[t]$. We say that
$P$ is \emph{bounded} if it is monic and if $Z$ does not have real roots. 
Notice that the motion~$\phi$ given by a bounded motion polynomial is a map
defined for all $t \in \R$; hence for such a motion we can use the notation
$\phi \colon \R \longrightarrow \p^3_{\R}$ instead of $\phi \colon \R
\dashrightarrow \p^3_{\R}$.
\end{definition}

\begin{remark}
 A bounded linear motion polynomial~$P$ represents a revolution around a point: 
if $P(t) = t - k$, then $t - k \in \K$ represents a rotation around a point~$Q$ 
that depends only on~$k$. This follows from Lemma~\ref{lemma:one_dim_motions}.
\end{remark}

\noindent We are now ready to state and prove the main result of this section:
\begin{theorem}
\label{theorem:bounded_factor}
Let $P \in \K[t]$ be a bounded motion polynomial. Then there exists a real polynomial
$R \in \R[t]$ such that $RP$ can be factored into linear polynomials.

If $P=Z+\eta W$ % has no non-constant real factors, 
and $\tilde{R} = \gcd(Z,\overline{Z})$, then the smallest such 
real polynomial is $R = \tilde{R} / \gcd(\tilde{R}, W\overline{W})$.
\end{theorem}
\begin{proof}
If $P$ is of the form $P=S\tilde{P}$ with $S\in\R[t]\setminus\R$ and
$\tilde{P}\in\K[t]$, then we can apply the theorem to $\tilde{P}$ obtaining
$R\in\R[t]$ such that $R\tilde{P}$ factors; thus also $RS\tilde{P}=RP$
factors.  Hence we may assume that $P$ does not have non-constant real
factors.

%Let $P(t) = Z(t) + \eta W(t)$ and let $\tilde{R} = \gcd(Z,\overline{Z})$. 
%The desired polynomial (of minimal
%degree) is obtained by $R = \tilde{R} / \gcd(\tilde{R}, W\overline{W})$.
Using Lemma~\ref{lemma:characterization_factorization} and
Proposition~\ref{proposition:factorization_ideal}, the proof is reduced to
showing that there exists a permutation of the $n$ roots of $RZ$ (counted with
multiplicities) such that the corresponding $G=\gcd(Q_1,\dotsc,Q_n)$
divides~$RW$. To this end write
\[
  Z = \prod_{i=1}^h \Bigl((t-\alpha_i)^{r_i} (t-\overline{\alpha_i})^{s_i} 
\Bigr),\quad
  r_i\geq s_i\geq0,\quad r_i>0,
\]
where $\alpha_i \in \C$ are pairwise conjugate-distinct, i.e.,
$\alpha_i\neq\alpha_j$ and $\alpha_i\neq\overline{\alpha_j}$ for $i\neq j$.
Note that in the special case $s_1=\dots=s_h=0$ we get $R=1$, in agreement
with the fact that in this case already $P$ itself factors by
Lemma~\ref{lemma:no_conj_roots}. On the other hand, we have always
$\deg(R)\leq\deg(P)$, and equality is attained when $r_i=s_i$ for all~$i$ and
$\gcd(Z,W)=1$.

Define $R_i:=(t-\alpha_i)(t-\overline{\alpha_i})$ and write
$\tilde{R}=\prod_{i=1}^h R_i^{s_i}$. Next let $u_i$ and $v_i$ be the
multiplicities of $\alpha_i$ and $\overline{\alpha_i}$ in~$W$, i.e.,
\[
  W = 
\tilde{W}\cdot\prod_{i=1}^h\Bigl((t-\alpha_i)^{u_i}(t-\overline{\alpha_i})^{v_i}
\Bigr)
\]
for some polynomial $\tilde{W}\in\C[t]$ such that $\gcd(\tilde{W},R_1\cdots
R_h)=1$. By introducing the quantity $m_i:=\min\{s_i,u_i+v_i\}$, the
polynomial~$R$ can be written as
\[
  R=\frac{\tilde{R}}{\gcd(\tilde{R}, W\overline{W})} = \prod_{i=1}^h 
R_i^{s_i-m_i},
\]
and hence
\begin{align*}
  RZ &= \prod_{i=1}^h\Bigl((t-\alpha_i)^{r_i+s_i-m_i} 
(t-\overline{\alpha_i})^{2s_i-m_i} \Bigr),\\
  RW &= \tilde{W}\cdot\prod_{i=1}^h\Bigl((t-\alpha_i)^{u_i+s_i-m_i} 
(t-\overline{\alpha_i})^{v_i+s_i-m_i} \Bigr).
\end{align*}
To construct an admissible permutation, the roots 
$\alpha_i,\overline{\alpha_i}$ are arranged as follows:
\[
\bigl(\underbrace{\overline{\alpha_i},\>\dotsc,\>\overline{\alpha_i}}_{s_i-\min\{
s_i,v_i\}\!\!},\>\underbrace{\alpha_i,\>\dotsc,\>\alpha_i}_{r_i+s_i-m_i},
\>\underbrace{\overline{ 
\alpha_i},\>\dotsc,\>\overline{\alpha_i}}_{s_i-\min\{s_i, u_i\}\!\!}\bigr).
\]
The assumption on~$P$ ensures that if $s_i>0$ then $u_i=0$ or $v_i=0$; indeed, if
$s_i,u_i,v_i$ were all strictly positive, then $R_i$ would be a common real
factor of $Z$ and~$W$.\linebreak Then $m_i=\min\{s_i,u_i\}+\min\{s_i,v_i\}$,
and hence the previous multiplicities agree with those in $RZ$.

The final permutation is obtained as an arbitrary interlacement of all
arrangements for $1\leq i\leq h$. Proposition~\ref{proposition:gcd} then implies that
\[
  G=\prod_{i=1}^h \Bigl( (t-\alpha_i)^{s_i-\min\{s_i,v_i\}} 
(t-\overline{\alpha_i})^{s_i-\min\{s_i,u_i\}} \Bigr).
\]
Since $s_i-\min\{s_i,v_i\}\leq u_i+s_i-m_i$ and $s_i-\min\{s_i,u_i\}\leq 
v_i+s_i-m_i$, we conclude that $G$ divides~$RW$.
\end{proof}
\begin{algorithm}
\caption{$\mathtt{FactorMotionPolynomial}$}
\begin{algorithmic}[1]
  \Require $P = Z + \eta \tth W \in \K[t]$ a bounded motion polynomial
    such that $Z$ and $W$ have no common factor in $\R[t]\setminus\R$.
  \Ensure a polynomial $R \in \R[t]$ and a tuple $(k_1, \dotsc, k_n)$ of 
    elements of $\K$ such that $(t-k_1) \cdots (t-k_n) = R(t) \cdot P(t)$.
  \Statex
  \State {\bfseries Factor} $Z(t)$ over $\C$, obtaining
    $Z = \prod_{i=1}^{h} (t - \alpha_i)^{r_i} (t - \overline{\alpha_i})^{s_i}$
    with $r_i\geq s_i\geq 0$.
  \State {\bfseries Initialize} $q = \mathtt{empty}$ (the empty tuple).
  \For {$i = 1, \dots, h$}
    \State {\bfseries Set} $u_i=\max_j\bigl((t-\alpha_i)^j\mid W\bigr)$ and
      $v_i=\max_j\bigl((t-\overline{\alpha_i})^j\mid W\bigr)$.
    \State {\bfseries Set} $m_i = \min \{ s_i, u_i + v_i \}$
    \State {\bfseries Set} $\omega =
\bigl(\underbrace{\overline{\alpha_i},\>\dotsc,\>\overline{\alpha_i}}_{s_i-\min\{
s_i,v_i\}\!\!},\>
      \underbrace{\alpha_i,\>\dotsc,\>\alpha_i}_{r_i+s_i-m_i},\>
\underbrace{\overline{\alpha_i},\>\dotsc,\>\overline{\alpha_i}}_{s_i-\min\{s_i,
u_i\}\!\!}\bigr).$
    \State {\bfseries Set} $q = \operatorname{\texttt{concatenate}}(q, \omega)$.
  \EndFor
  \State {\bfseries Set} $R = \prod_{i = 1}^{h} \bigl((t - \alpha_i)(t - 
\overline{\alpha_i})\bigr)^{s_i - m_i}$.
  \State {\bfseries Set} $n = \operatorname{\texttt{length}}(q)$.
  \State {\bfseries Set} $Q_j = \prod_{l = 1}^{j-1} (t - q_{l}) \prod_{l =
    j+1}^{n} (t - \overline{q_l})$ for all $j \in \{ 1, \dotsc, n \}$.
  \State {\bfseries Compute} $\{ w_j \}_{j = 1}^{n}$ such that
    $R \tth W = \sum_{j=1}^{n} w_j Q_j$ using linear algebra.
  \State {\bfseries Set} $k_j = q_j - \eta \tth w_j$ for all 
$j\in\{1,\dotsc,n\}$.
  \State \Return $\bigl( R, (k_1, \ldots, k_n) \bigr)$.
\end{algorithmic}
\end{algorithm}

The proof of Theorem~\ref{theorem:bounded_factor} immediately gives
rise to a factorization algorithm for motion polynomials, called 
\texttt{FactorMotionPolynomial}. It produces one single factorization 
for a given polynomial~$P$, although there could be many for the 
following reasons: first, there may be several admissible permutations of the
roots of~$RZ$, and second, the choice of the coefficients $\{ w_j \}$ in 
Step~$12$ may not be unique, since the polynomials~$\{ Q_j \}$ need not be 
$\C$-linearly independent.

\begin{example}
\label{example:elliptic_motion_reviewed}
We can review Example~\ref{example:elliptic_motion} in the light of 
Theorem~\ref{theorem:bounded_factor} and the algorithm 
\texttt{FactorMotionPolynomial}. We consider a particular instance of the 
collection of motion polynomials
\[
  P(t) \; = \; (t^2 + 1) + \eta \tth (at - b\imath),
\]
namely we choose $a = \imath$ and $b = 2\imath$, obtaining the motion we 
presented in Section~\ref{overview}. If we run 
\texttt{FactorMotionPolynomial} with $P$ as input, we obtain $R = t^2 + 1$, and 
if we fix the permutation $(\imath,-\imath,-\imath,\imath)$ of the roots 
of~$\ppt(RP)$, as in~\eqref{equation:ansatz_elliptic}, we obtain a 
two-dimensional family of factorizations into linear factors, i.e., the 
solutions of the linear system~\eqref{equation:solution_elliptic}, out of which 
we choose the factorization used in Section~\ref{overview}.
\end{example}

\section{The flip procedure}
\label{flip}

The procedure described in this section is crucial for the construction
of linkages with mobility one. It is inspired by similar techniques --- involving 
the interchange of factors in a factorization of a quadratic motion 
polynomial --- used in~\cite{HegedusSchichoSchroecker2012} 
and~\cite{HegedusSchichoSchroecker2013a}. 

\begin{lemma}
\label{lemma:flip}
Let $k_1, k_2 \in \K$ be such that $\ppt(k_1) \neq \overline{\ppt(k_2)}$. 
Then there exists a unique pair $(k_3, k_4) \in \K^2$ such that:
\begin{enumerate}
  \item\label{lemma:flip:case1} $\ppt(k_3) = \ppt(k_2)$ and $\ppt(k_4) = \ppt(k_1)$;
  \item\label{lemma:flip:case2} $(t - k_1)(t - k_2) = (t - k_3)(t - k_4)$ as polynomials in $\K[t]$.
\end{enumerate}
\end{lemma}
\begin{proof}
Let us suppose that $k_i = z_i + \eta \tth w_i$ for $i \in \{1,2\}$. We make 
the ansatz $k_3 = z_2 + \eta \tth w_3$ and $k_4 = z_1 + \eta \tth w_4$,
where $w_3$ and $w_4$ are elements of $\C$ to be determined. Then
condition~\ref{lemma:flip:case2} is equivalent to the linear system 
\[
  \begin{pmatrix} 1 & 1 \\ z_1 & \overline{z}_2 \end{pmatrix}
  \begin{pmatrix}	w_3 \\ w_4 \end{pmatrix} =
  \begin{pmatrix}	w_1 + w_2 \\ \overline{z}_1 w_2 + z_2 w_1 \end{pmatrix},
\]
which has a unique solution, since by hypothesis the determinant of the
matrix is different from zero.
\end{proof}

\begin{definition}
\label{definition:flip}
	Let $k_1, k_2 \in \K$ be such that $\ppt(k_1) \neq \overline{\ppt(k_2)}$. 
Then we define $\Flip(k_1, k_2) := (k_3, k_4)$, where $(k_3, k_4) \in 
\K^2$ is the pair from Lemma~\ref{lemma:flip}.
\end{definition}

\begin{definition}
Let $k = z + \eta \tth w \in \K$ with $z\in\C\setminus\R$. Then $k$
represents a rotation around a point that we denote by~$\midpt{k}$. Using 
Equation~\eqref{equation:action_K_C} one obtains that
\[ 
  \midpt{k} \; = \; \frac{w}{\overline{z} - z}.
\]
\end{definition}

\begin{definition}
\label{definition:flip_linkage}
Let $k_1, k_2 \in \K$ be such that $\ppt(k_1) \neq \overline{\ppt(k_2)}$ and
$\ppt(k_1),\ppt(k_2)\in\C\setminus\R$.  We define the \emph{flip linkage} 
of $k_1,k_2$ via the algorithm \texttt{FlipLinkage} below.
The graph of this flip linkage is shown in Figure~\ref{figure:mobility}.
	\begin{algorithm}
	\caption{$\mathtt{FlipLinkage}$}
		\begin{algorithmic}[1]
		\Require $k_1, k_2 \in \K$ such that $\ppt(k_1) \neq \overline{\ppt(k_2)}$
                  and $\ppt(k_1),\ppt(k_2)\in\C\setminus\R$.
		\Ensure $(G,\rho)$ a flip linkage.
		\Statex
		\State {\bfseries Define} $V = \{ 1, 2, 3, 4 \}$ and $E = \{ (1,2), (2,3), 
(3, 4), (1, 4) \}$ and set $G = (V, E)$.
		\State {\bfseries Define} $(k_3, k_4) = \Flip(k_1, k_2)$.
		\State {\bfseries Define} $u_i = \midpt{k_i}$ for all $i \in \{ 1, \dotsc, 
4 \}$.
		\State {\bfseries Define}
		\[
			\begin{array}{ccccccc}
				\rho(1,3) & = & \bigl(\Re(u_1),\Im(u_1)\bigr) & & \rho(2,1) & = &
\bigl(\Re(u_2),\Im(u_2)\bigr) \\
				\rho(3,4) & = & \bigl(\Re(u_3),\Im(u_3)\bigr) & & \rho(4,2) & = &
\bigl(\Re(u_4),\Im(u_4)\bigr)
			\end{array}
		\]
		\State \Return $(G,\rho)$.
		\end{algorithmic}
	\end{algorithm}
\end{definition}
\begin{lemma}
\label{lemma:same_length}
	Let $k_1, k_2 \in \K$ with $\overline{\ppt(k_1)} \neq \ppt(k_2)$. Let 
$(k_3, k_4) = \Flip(k_1, k_2)$. Then 
	\[ \bigl| \midpt{k_1} - \midpt{k_2} \bigr| \; = \; \bigl| \midpt{k_3} - 
\midpt{k_4} \bigr|.  \]
If, in addition, $z_1, z_2 \not\in \R$ and $\midpt{k_1} \neq \midpt{k_2}$, then 
the following holds:
	\[ \bigl| \midpt{k_1} - \midpt{k_2} \bigr| \; \neq \; \bigl| \midpt{k_1} - 
\midpt{k_3} \bigr|.  \]
\end{lemma}
\begin{proof}
Let $k_1 = z_1 + \eta \, w_1$ and $k_2 = z_2 + \eta \, w_2$. Then the
condition $(t - k_1)(t - k_2) = (t - k_3)(t - k_4)$ implies that
$k_3 = z_2 + \eta \tth w_3$ and $k_4 = z_1 + \eta \tth w_4$ with
\[
  w_3 = \frac{(z_2-\overline{z_2}) \, w_1 + (\overline{z_1}-\overline{z_2}) \, 
w_2}{z_1-\overline{z_2}}
  \quad\text{and}\quad
  w_4 = \frac{(z_1-z_2) \, w_1 + (z_1-\overline{z_1}) \, 
w_2}{z_1-\overline{z_2}}.
\]
A direct calculation shows that
\[
  \bigl| \midpt{k_3} - \midpt{k_4} \bigr| =
  \left| \frac{w_3}{\overline{z_2}-z_2} - \frac{w_4}{\overline{z_1}-z_1} 
\right| =
  \left| \frac{w_1}{\overline{z_1}-z_1} - \frac{w_2}{\overline{z_2}-z_2} 
\right| =
  \bigl| \midpt{k_1} - \midpt{k_2} \bigr|
\]
as claimed. The second claim is also obtained via a direct computation. 
\end{proof}
Lemma~\ref{lemma:same_length} implies that the linkages produced by the
algorithm \texttt{FlipLinkage} have the shape of a (possibly degenerated)
antiparallelogram, as in the right part of Figure~\ref{figure:weak_strong}. 
We introduce the notion of \emph{flip mobility} to exclude some degenerated 
cases, namely the ones where the antiparallelogram is in fact a square.

\begin{notation}
\label{notation:flip_mobility}
	Let  $k_1 = z_1 + \eta \tth w_1$ and $k_2 = z_2 + \eta \tth w_2$ be two 
elements of $\K$. We say that the condition $\FM{k_1}{k_2}$ (which stands for 
\emph{flip mobility}) holds if and only if
	\[ z_1, z_2 \not\in \R, \qquad z_1 \neq z_2, \qquad z_1 \neq \overline{z}_2,
\qquad \midpt{k_1} \neq \midpt{k_2}. \]
\end{notation}

\begin{lemma}
\label{lemma:flip_mobility}
	Let $k_1, k_2 \in \K$ such that $\FM{k_1}{k_2}$ holds. Then the linkage 
$L$ obtained by $\mathtt{FlipLinkage}(k_1, k_2)$ has mobility one and the 
configuration curve has two components~$C_1$ and~$C_2$. Moreover, the natural 
maps $C_{1} \longrightarrow \vrpos{i}{j}$ and $C_{2} \longrightarrow 
\vrpos{i}{j}$ 
are isomorphisms for every two neighboring links~$i$ and~$j$. 
\end{lemma}
\begin{proof}
Since the linkages obtained via $\mathtt{FlipLinkage}$ are constituted by 
four bars, this is a well-known result; see for example~\cite[Table on 
p.~426]{Bottema1979}. One can also prove the statement using the notions we 
introduced in Section~\ref{linkages} as follows. We set $L = 
\mathtt{FlipLinkage}(k_1, k_2)$, and denote $(k_3, k_4) = \Flip(k_1, k_2)$. 
Recalling Definition~\ref{definition:mobility}, our goal is to prove that 
$\conf{L}$ is a one-dimensional variety. One notices that
\[ 
  \conf{L} \; \cong \; \bigl\{ (\lambda_1:\mu_1), \dotsc, 
  (\lambda_4:\mu_4) \; : \; f_{12}(\lambda_1,\mu_1,\lambda_2,\mu_2) = 
  f_{34}(\lambda_3,\mu_3,\lambda_4,\mu_4) \bigr\},
\]
where
\[
  \begin{array}{rccc}
    f_{12}\colon & \p^1_{\R} \times \p^1_{\R} & \longrightarrow & \p(\K) \\
    & (\lambda_1: \mu_1), (\lambda_1: \mu_1) & \mapsto & (\lambda_1 - \mu_1 
    \tth k_1)(\lambda_2 - \mu_2 k_2)
  \end{array}
\]
and analogously for $f_{34}$. The images of~$f_{12}$ and~$f_{34}$ are two 
smooth quadrics whose intersection, equal to~$\conf{L}$, is a curve with two 
components. Using this it is possible to prove the claim.
\end{proof}

\begin{remark}
	Results in Example~\ref{example:mobility} are particular instances of 
Lemma~\ref{lemma:flip_mobility}. One can in fact check that in that case $L = 
\mathtt{FlipLinkage}(k_1, k_2)$, where $k_1 = \imath - \eta \cdot 3$ and $k_2 = 
-2\imath - \eta \tth (2 + \imath)$.
\end{remark}

\section{Construction of linkages}
\label{construction}

In this section, we show how to construct a linkage with mobility one that
traces an algebraic curve described by a rational parametrization. Based on
the results of Section~\ref{factorization}, we first describe an algorithm
which takes a bounded motion polynomial~$P$ and constructs a linkage
weakly realizing~$P$ (see Definition~\ref{definition:realizing}).
\begin{algorithm}
\caption{$\mathtt{ConstructWeakLinkage}$}
	\begin{algorithmic}[1]
		\Require $P \in \K[t]$ a bounded motion polynomial.
		\Ensure $(G, \rho)$ a linkage weakly realizing~$P$.
		\Statex
                \State {\bfseries Compute} $S\in\R[t]$ of maximal degree, monic,
such that $S$ divides $P$.
		\State {\bfseries Compute} $\bigl(R, (k_1, \dotsc, k_n)\bigr) = 
\mathtt{FactorMotionPolynomial}(P / S)$.
		\State {\bfseries Set} $u_i = \midpt{k_i} = \frac{w_i}{\overline{z}_i - 
z_i}$, where $k_i = z_i + \eta \tth w_i$, for every $i \in \{1, \dotsc, n\}$. 
		\State {\bfseries Set} $V = \{ 1, \dotsc, n+1 \}$ and $E = \{ (1,2), (2,3), 
\dotsc, (n, n+1) \}$. 
		\State {\bfseries Set} $G = (V,E)$.
		\State {\bfseries Set} $\rho(i,i+1) = \bigl(\Re{(u_i)}, \Im{(u_i)}\bigr) \in 
\R^2$ for every $i \in \{1, \dotsc, n\}$.
		\State \Return $(G, \rho)$.
	\end{algorithmic}
\end{algorithm}
\begin{proposition}
\label{proposition:construct_weak}
	Algorithm~\texttt{ConstructWeakLinkage} is correct.
\end{proposition}
\begin{proof}
	The linkage~$L$ constructed by the algorithm \texttt{ConstructWeakLinkage} is
an open chain, so its configuration space is isomorphic to
$\left(\p_{\C}^{1}\right)\!{}^{n-1}$, where $n-1$ is the number of joints 
of~$L$. Let $\phi\colon \R \longrightarrow \p_{\C}^3$ be the motion 
corresponding to $R P$. We consider links $1$ and $n$, and we prove that 
$\mathrm{image}(\phi) \subseteq \rpos{1}{n}$, namely that $L$ weakly realizes
$RP$ according to Definition~\ref{definition:realizing}. In our situation we
have that
\[ 
  \rpos{1}{n} = \bigl\{ \sigma_{1,2} \circ \sigma_{2,3} \circ \cdots \circ
\sigma_{n-1,n} \; : \: \sigma_{i,i+1} \in \p_{\C}^{1} \bigr\}. 
\]
Now we fix an arbitrary $t \in \R$ and we take $\sigma_{i,i+1}$ to be the
isometry given by $t - k_i$ for all $i \in \{1, \ldots, n-1\}$. Then
$\phi(t) \in \rpos{1}{n}$, so the claim is proven.
\end{proof}

Suppose we are given a rational parametrization of a bounded curve. As
we saw in Proposition~\ref{prop:curve_motion}, starting from it we can
construct a bounded polynomial~$P$ such that the curve is the orbit of
the origin in~$\R^2$ under the motion described by~$P$. Then the algorithm
$\mathtt{ConstructWeakLinkage}$ returns a linkage~$L$ that weakly realizes~$P$. 
However, this is rather unsatisfactory, since the configuration space of~$L$
is, in general, a variety of high dimension. Our goal is to obtain a linkage 
with mobility one, providing a strong realization of~$P$. The main idea is to 
``rigidify''~$L$, which by construction is an open chain, by introducing 
additional links and joints, forming antiparallelograms. At the level of
graphs, this corresponds to extending the linear graph of~$L$ to a $2\times n$ 
ladder graph (see Figure~\ref{figure:strong_linkage}). To do so, we pick a 
linear motion polynomial $t-l$, which will constitute the first ``step'' of the 
ladder, and then apply the flip procedure iteratively to complete the 
ladder. In the following, we investigate how $l \in \K$ has to be chosen so
that the resulting linkage has mobility~one.

\begin{definition}
\label{definition:iterated_flip_mobility}
	Let $\mathbf{k} = (k_1, \dotsc, k_m)$ be a tuple in $\K^m$, 
and let $l \in \K$. Define
	\[ l_1 \; := \; l, \qquad (\tilde{k}_i,l_{i+1}) \; := \; \Flip(l_i, 
k_i) \qquad \text{for all }i\in \{1, \dotsc, m \}. \]
Then we say that the condition $\IFM{l}{\mathbf{k}}$ (which stands for 
\emph{iterated flip mobility}) holds if and only if $\FM{l_i}{k_i}$ holds for
all $i \in \{1, \dotsc, m \}$.
\end{definition}

In order to prove that the algorithm we propose works, we have to show 
that it is always possible to find an element~$l \in \K$ such that the iterated 
flip mobility condition holds. For this we use a property of flip linkages.

\begin{definition}
\label{definition:inv}
Given an element $k = z + \eta \tth w \in \K$ with $z \neq 0$, we define
$\inv(k)$ to be the element $\overline{z} - \eta \tth w \in \K$. Recalling
Remark~\ref{remark:algebra_inverses}, we have that~$k$
represents an isometry~$\sigma$, and $\inv(k)$ its inverse~$\sigma^{-1}$.
Moreover, if $z \not\in \R$, then $t - k$ represents an isometry for every
$t \in \R$, and $t - \inv(k)$ the inverse isometry.
\end{definition}

\begin{lemma}
\label{lemma:revert_flip}
	Let $k_1, k_2 \in \K$ be such that $\ppt(k_1) \neq \overline{\ppt(k_2)}$ and 
$\ppt(k_2) \not\in \R$. Let $(k_3, k_4) = \Flip(k_1, k_2)$. Then we have that
\[ \bigl( \inv(k_3), k_1 \bigr) = \Flip \bigl( k_4, \inv(k_2) \bigr). \]
\end{lemma}
\begin{proof}
	Notice that the fact that $\ppt(k_2) \not\in \R$ implies that the element $(t 
- k_2)$ represents an isometry for all $t \in \R$. From 
Definition~\ref{definition:inv} one has that the inverse isometry is
represented by $\bigl(t - \inv(k_2)\bigr)$. The two situations are depicted 
below: 
\[
	\begin{array}{ccc}
		\xymatrix@C=1.2cm@R=1.2cm{\bullet \ar[r]^{t - k_2} & \bullet \\ \bullet 
\ar[u]^{t - k_1} \ar[r]_{t - k_3} & \bullet \ar[u]_{t - k_4}} & & 
\xymatrix@C=1.2cm@R=1.2cm{\bullet & \bullet \ar[l]_{t - \inv(k_2)} \\ \bullet 
\ar[u]^{t - k_1} & \bullet \ar[l]^{t - \inv(k_3)} \ar[u]_{t - k_4}} \\
		(k_3, k_4) = \Flip(k_1, k_2) & \quad & \bigl( \inv(k_3), k_1 \bigr) = \Flip 
\bigl( k_4, \inv(k_2) \bigr)
	\end{array}
\]
A direct computation shows the desired result.
\end{proof}

\begin{lemma}
\label{lemma:IFM_holds}
 Let $\mathbf{k} = (k_1, \dotsc, k_m)$ be a tuple in $\K^m$ such that 
$\ppt(k_i) \not\in \R$ for all $i \in \{1, \dotsc, m\}$. Then there exists $l 
\in \K$ with $\ppt(l) \not\in \R$ such that $\IFM{l}{\mathbf{k}}$ holds.
\end{lemma}
\begin{proof}
Following the notation of Definition~\ref{definition:iterated_flip_mobility} we
have that $\ppt(l_i) = \ppt(l)$ for all $i \in \{1, \dotsc, m \}$. Hence, in
order to ensure that $\IFM{l}{\mathbf{k}}$ holds, we need to choose~$l$ such 
that $\ppt(l) \neq \ppt(k_i)$ and $\ppt(l) \neq \overline{\ppt(k_i)}$ for all $i 
\in \{1, \dotsc, m \}$. From now on, we fix such a value for~$\ppt(l)$. The 
other situations we should keep away from are the ones where $\midpt{l_i} =
\midpt{k_i}$. Since $\ppt(l_i)$ is now fixed, the condition $\midpt{l_i} =
\midpt{k_i}$ becomes a linear equation for $\spt(l_i)$, which always admits a
unique solution $\hat{w_i} \in \C$. Let us define $\hat{l_i} = \ppt(l) + \eta\,
\hat{w_i}$. Using Lemma~\ref{lemma:revert_flip} repeatedly (namely $i-1$ times)
starting from the pair $\bigl(\hat{l_i}, \inv(k_{i-1})\bigr)$, we obtain an
element $h_i \in \K$ such that
\[ 
  \midpt{l_i} \neq \midpt{k_i} \quad \iff \quad l \neq h_i. 
\]
Summing up, for every fixed value of $\ppt(l)$ different from all 
$\bigl\{ \ppt(k_i) \bigr\}_{i=1}^{m}$ and all $\bigl\{ \overline{\ppt(k_i)} 
\bigr\}_{i=1}^{m}$, there are finitely many values for $\spt(l)$ that should be
avoided for $\IFM{l}{\mathbf{k}}$ to hold. Thus we can always find $l \in
\K$ with the claimed properties.
\end{proof}

We are ready to present the algorithm \texttt{ConstructStrongLinkage}. 
Figure~\ref{figure:strong_linkage} shows the labeled graph produced by the 
algorithm.

\begin{algorithm}
\caption{$\mathtt{ConstructStrongLinkage}$}
	\begin{algorithmic}[1]
		\Require $P \in \K[t]$ a bounded motion polynomial.
		\Ensure $(G, \rho)$ a linkage strongly realizing the motion induced by~$P$.
		\Statex
                \State {\bfseries Compute} $S\in\R[t]$ of maximal degree, monic,
such that $S$ divides $P$.
		\State {\bfseries Compute} $\bigl(R, \mathbf{k} = (k_1, \dotsc, k_n)\bigr) =
\mathtt{FactorMotionPolynomial}(P / S)$.
		\State {\bfseries Choose} $l \in \K$ such that $\IFM{l}{\mathbf{k}}$ 
holds.
		\State {\bfseries Set} $l_1 = l$ and $(\tilde{k}_i,l_{i+1}) \; = \; 
\Flip(l_i, k_i)$ for all $i \in \{1, \dotsc, n \}$.
		\State {\bfseries Set} 
		\[ \left. \begin{array}{rcl} 
			u_i & = & \midpt{k_i} \\
			\tilde{u}_i & = & \midpt{\tilde{k}_i}  \\
			v_j & = & \midpt{l_j}
			\end{array} \right\} \;\; 
			\begin{array}{l}
				\mathrm{for\ every\ } i \in \{1, \dotsc, n\} \\
				\mathrm{for\ every\ } j \in \{1, \dotsc, n+1\}
			\end{array} \]
		\State {\bfseries Set} $V = \{ 1, \dotsc, 2n + 2 \}$.
		\State {\bfseries Set} $E = \Bigl\{ (i,i+1), (n+1+i, n+2+i), (j, n+1+j) 
\mathrm{\ for\ all\ } i \in \{ 1, \dotsc, n \} \mathrm{\ and\ } j \in \{ 1,
\dotsc, n+1 \} \Bigr\}$.
		\State {\bfseries Set} $G = (V,E)$.
		\State {\bfseries Set} 
			\[ \left. \begin{array}{rcl}
				\rho(i, i+1) & = & \bigl(\Re{(u_i)}, \Im{(u_i)}\bigr) \\ 
				\rho(n+1+i, n+2+i) & = & \bigl(\Re{(\tilde{u}_i)}, 
\Im{(\tilde{u}_i)}\bigr) \\
				\rho(j, n+1+j) & = & \bigl(\Re{(v_j)}, \Im{(v_j)}\bigr)
			\end{array} \right\} \;\; 
			\begin{array}{l}
				\mathrm{for\ every\ } i \in \{1, \dotsc, n\} \\
				\mathrm{for\ every\ } j \in \{1, \dotsc, n+1\}
			\end{array}  \]
		\State \Return $(G, \rho)$.
	\end{algorithmic}
\end{algorithm}

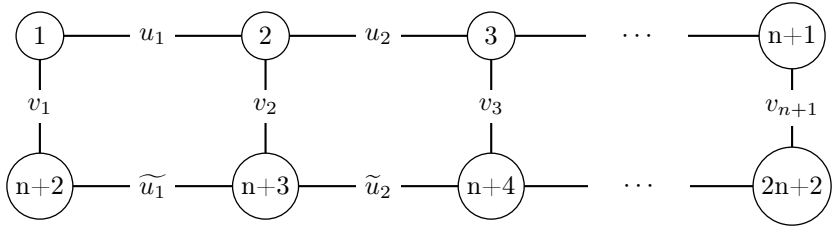
\begin{figure}[ht]
	\begin{tikzpicture}[thick]
		\Vertex[x=0,y=2]{1}
		\Vertex[x=0,y=0]{n+2}
		\Vertex[x=3,y=0]{n+3}
		\Vertex[x=3,y=2]{2}
		\Vertex[x=6,y=2]{3}
		\Vertex[x=6,y=0]{n+4}
		\Vertex[x=10,y=2]{n+1}
		\Vertex[x=10,y=0]{2n+2}
		\Edge[label=$u_1$](1)(2)
		\Edge[label=$v_1$](1)(n+2)
		\Edge[label=$\tilde{u_1}$](n+2)(n+3)
		\Edge[label=$v_2$](2)(n+3)
		\Edge[label=$u_2$](2)(3)
		\Edge[label=$\tilde{u}_2$](n+3)(n+4)
		\Edge[label=$v_3$](3)(n+4)
		\Edge[label=$v_{n+1}$](n+1)(2n+2)
		\Edge[label=$\quad\cdots\quad$](3)(n+1)
		\Edge[label=$\quad\cdots\quad$](n+4)(2n+2)
	\end{tikzpicture}
	\caption{The labeled graph of the linkage returned by the algorithm
\texttt{ConstructStrongLinkage}.}
	\label{figure:strong_linkage}
\end{figure}

\begin{theorem}
\label{theorem:strong_realization}
Let~$P$ be a bounded polynomial in~$\K[t]$ and let $\phi\colon \R
\longrightarrow \p^3_{\R}$ be the rational motion induced by~$P$. Let us denote 
by~$L$ the linkage obtained by applying \texttt{ConstructStrongLinkage} to~$P$. 
Then $L$ strongly realizes the motion~$\phi$.
\end{theorem}
\begin{proof}
Let~$K$ be the configuration space of~$L = (G, \rho)$. As in the statement of 
Lemma \ref{lemma:flip_mobility}, we are going to prove that for any two 
neighboring edges $i$ and $j$ of $G$, every irreducible component of~$K$ is 
isomorphic to~$\rpos{i}{j}$. This will prove that~$K$ has dimension one, which 
means that~$L$ has mobility one.

Recall from Definition~\ref{definition:configuration} that $K$ is the 
collection of configurations $\Sigma = ( \sigma_{k,l} )$ satisfying the 
equations imposed by the directed cycles of $G$.
First of all, notice that one can consider only those equations imposed by 
directed cycles of the form:
\[
  \xymatrix@R=1cm@C=1.5cm{\mathllap{h \;\; \bullet} \ar[r] & 
\mathrlap{\bullet \;\; h+1} \ar[d] \\ 
\mathllap{(n+1)+h \;\; \bullet} \ar[u] & 
\mathrlap{\bullet \;\;(n+1)+(h+1)} \ar[l] }
\]
In fact, all other equations belong to the ideal generated by 
these ones. We introduce the following notation: for $h \in \{1, \dotsc, 
n\}$, denote by~$\mathrm{sq}(h)$ the set 
\[
  \bigl\{ \bigl(h,h+1\bigr), \bigl(h+1, (n+1)+(h+1)\bigr), \bigl((n+1)+(h+1), 
  (n+1)+h\bigr), \bigl((n+1)+h, h\bigr) \bigr\}.
\]
We define~$\pi_{\mathrm{sq}(h)}$ to be the projection
\[
  \pi_{\mathrm{sq}(h)}\colon K \longrightarrow \prod_{(r,s) \in \mathrm{sq}(h)} 
  \vrpos{r}{s}.
\]
Then the image of~$\pi_{\mathrm{sq}(h)}$ is contained in the configuration 
space~$C_h$ of a flip linkage, which by Lemma~\ref{lemma:flip_mobility} has two 
components, so $C_h = C_h^1 \cup C_h^2$. 

For every sequence~$b \in \{1,2\}^n$, we set
\[
  K_{b} \; := \; \pi_{\mathrm{sq}(1)}^{-1} \left(C_1^{b(1)} \right) \cap \dotsb 
\cap \pi_{\mathrm{sq}(n)}^{-1} \left(C_n^{b(n)} \right).
\]
By construction we have $K = \bigcup_{b \in \{1,2\}^n} K_b$. We are going 
to prove that for every sequence $b$ and for any two neighboring links~$i$ 
and~$j$ in~$G$, we have $K_{b} \cong \vrpos{i}{j}$.

For the sake of simplicity, we start by proving this result when the graph~$G$ 
consists only of two squares (i.e.\ $n=2$):
\[
  \begin{tikzpicture}[thick]
    \Vertex[x=0,y=1.6]{1}
    \Vertex[x=0,y=0]{4}
    \Vertex[x=3,y=0]{5}
    \Vertex[x=3,y=1.6]{2}
    \Vertex[x=6,y=1.6]{3}
    \Vertex[x=6,y=0]{6}
    \Edge(1)(2)
    \Edge(2)(3)
    \Edge(1)(4)
    \Edge(2)(5)
    \Edge(3)(6)
    \Edge(4)(5)
    \Edge(5)(6)
  \end{tikzpicture}
\]
In this case we only have two flip linkages with configuration spaces~$C_1$ 
and~$C_2$. Let us fix~$b \in \{1,2\}^2$, and consider the diagram:
\begin{equation}
\label{equation:pullback}
  \begin{array}{c}
  \xymatrix{C_1^{b(1)} \ar[rd]^{\cong}_{\pi_{2,5}} & & C_2^{b(2)} 
\ar[ld]_{\cong}^{\pi_{2,5}} \\ & \vrpos{2}{5}}
  \end{array}
\end{equation}
The maps~$\pi_{2,5}$ are isomorphisms, because of 
Lemma~\ref{lemma:flip_mobility}. We consider the set
\[
  \Bigl\{ (x,y) \; : \; x \in C_1^{b(1)}, y \in C_2^{b(2)} \mathrm{\ and\ 
} \pi_{2,5}(x) = \pi_{2,5}(y) \Bigr\},
\]
which is nothing but the pullback of Diagram~\eqref{equation:pullback}. One can 
check that this set equals~$K_{b}$ since, as we observed before, the 
elements of~$K$ are only subject to the equations coming from the cycles 
in~$C_1$ and~$C_2$. Since isomorphisms are stable under pullbacks, we obtain the
following commutative diagram of isomorphisms:
\[
  \xymatrix{& K_b \ar[rd]_{\cong}^{\pi_{\mathrm{sq}(2)}} 
\ar[ld]^{\cong}_{\pi_{\mathrm{sq}(1)}} \\ C_1^{b(1)} 
\ar[rd]^{\cong}_{\pi_{2,5}} 
& & C_2^{b(2)} 
\ar[ld]_{\cong}^{\pi_{2,5}} \\ & \vrpos{2}{5}}
\]
The composition $\pi_{2,5} \circ \pi_{\mathrm{sq}(1)}$ equals the 
projection $\pi_{2,5}\colon K_{b} \longrightarrow \vrpos{2}{5}$, and the same 
holds for every projection $\pi_{i,j}\colon K_{b} \longrightarrow 
\vrpos{i}{j}$, where~$i$ and~$j$ are neighboring links. So each of these 
maps is an isomorphism, proving our claim. 

If now $G$ is constituted by more than two squares, then $K_{b}$ is obtained 
via an iteration of several pullbacks, and eventually we get a diagram of 
isomorphisms. Thus the claim holds also in this case. Since $K$ is covered by 
finitely many varieties of dimension one, it has dimension one. Moreover, 
each of the varieties~$K_{b}$ is an irreducible component of~$K$.

We are left to prove that $L$ realizes the motion~$\phi$ induced by~$P$. 
Similarly as we did in the proof of
Proposition~\ref{proposition:construct_weak}, one can show that
$\mathrm{image}(\phi) \subseteq \rpos{1}{n+1}$. Since $K$ is one-dimensional,
also $\rpos{1}{n+1}$ is so; since $\mathrm{image}(\phi)$ is one-dimen\-sional,
$\overline{\mathrm{image}(\phi)}$ is a component of~$\rpos{1}{n+1}$. 
This concludes the proof.
\end{proof}

\begin{remark}
\label{remark:bracing}
	As we mentioned in Section~\ref{introduction}, the linkages constructed via 
our algorithms present the same issues as the ones produced by Kempe's 
procedure, namely the devices do not realize only the motion they are designed 
for, but also other ones. This is due to the fact that the elementary linkages 
produced by the flip procedure, commonly known as antiparallelograms, admit 
configuration spaces with more than one component (see
Lemma~\ref{lemma:flip_mobility}). On the other hand, since 
antiparallelograms are among the linkages used also by Kempe's procedure, the 
techniques developed in~\cite{Abbott2008} 
and~\cite[Section~3.2.2]{DemaineORourke2007} can be applied also in our case to 
prevent this unwanted behavior.
\end{remark}

\begin{remark}
Let $\phi$ be a rational motion given by $Z+\eta\,W\in\K[t]$.  Our algorithm
allows one to construct a linkage that generates~$\phi$, using a motor that
rotates at constant speed. For this purpose one needs two neighboring links
whose relative position is a revolution with constant speed; such a motion is
represented by a motion polynomial of the form $t\pm\imath+\eta\,w$. If
$\imath$ or $-\imath$ is already a root of $Z$, then any linkage returned by
\texttt{ConstructStrongLinkage} has the desired property. If not, then the
linear polynomial~$l$ in step~3 of this algorithm has to be chosen
accordingly.
\end{remark}

We now investigate how many links and joints the linkages created by our
algorithm have. We consider the tasks of strongly realizing a given rational
motion, and of drawing a given parametrized curve.
\begin{proposition}
\label{prop:realize}
Let $\phi$ be a rational motion given by a motion polynomial~$P$ of degree~$d$.
Then there exists a linkage with at most $4d+2$ links and $6d+1$ strongly
realizing~$\phi$.
\end{proposition}
\begin{proof}
Apply \texttt{ConstructStrongLinkage} to~$P$, and notice that the real
polynomial~$R$ produced by \texttt{FactorMotionPolynomial} has degree
at most~$d$.
\end{proof}

\begin{proposition}
\label{prop:draw}
Let $\varphi\colon \R \longrightarrow \R^2$ be the parametrization of a real
bounded planar curve. Without loss of generality, suppose that
\[
  \varphi(t) \; = \; \left( \frac{f(t)}{h(t)}, \frac{g(t)}{h(t)} \right)
\]
with $f,g,h \in \R[t]$, and $h$ monic, and $d := \deg{h} > \max \{ \deg{f},
\deg{g} \}$. There exists a linkage with at most $3d+2$ links and 
$\frac{9}{2} \tth d + 1$ joints drawing the curve given by~$\varphi$.
\end{proposition}
\begin{proof}
As discussed before Definition~\ref{definition:bounded_poly}, the motion
polynomial $P = h + \eta \tth (f + \imath g)$ is bounded, and by
Proposition~\ref{prop:curve_motion} the orbit of the origin in $\R^2$ under
the motion determined by~$P$ is the image of~$\varphi$.  We could apply
\texttt{ConstructStrongLinkage} to~$P$, but we obtain a smaller linkage by
applying this algorithm to $P'=CP$, where $C\in\C[t]$ is any polynomial such
that $C\overline{C}=h$ and $\gcd(C,\overline{C})=1$; such a $C$ exists because
$h$ has no real roots.  In fact, as discussed in
Remark~\ref{remark:multiply_complex}, the orbit of the origin under the motion
given by $P'$ is also the image of~$\varphi$. Moreover, by
Theorem~\ref{theorem:bounded_factor}, a motion polynomial $Z+\eta\,W$ is
factorizable if the corresponding real polynomial~$R$ equals~$1$. This is the
case if and only if $\gcd(Z,\overline{Z}) \mid W\overline{W}$; it is easy to
check that this condition holds for~$P'$. Finally, the linkage obtained by
applying \texttt{ConstructStrongLinkage} to~$P'$ is constituted by $3d + 2$
links and $\frac{9}{2} \tth d + 1$ joints, because the degree of~$P'$ is
$\frac{3}{2} \tth d$ --- note that $d$ must be even.
\end{proof}

\begin{example}
\label{example:elliptic_motion_flips}
We illustrate the algorithm \texttt{ConstructStrongLinkage} with the example
from Section~\ref{overview}. In contrast to Example~\ref{example:elliptic_motion_reviewed},
here we multiply the polynomial $P(t)=(t^2 + 1) + \eta \tth (\imath t - 2)$ by
$C=t-\imath$, as prescribed by Proposition~\ref{prop:draw}. Calling
\texttt{FactorMotionPolynomial} with $CP$ as input, we get
$R=1$ and obtain:
\[
  C(t)\cdot P(t) \; = \; \bigl(t - \imath - \tfrac12 \eta\, \imath\bigr)\cdot
       \bigl(t - \imath + \tfrac12 \eta\, \imath\bigr)\cdot
       \bigl(t + \imath + \eta\, \imath\bigr).
\]
Next we have to fix an~$l\in\K$ such that the IFM condition holds; it turns
out that $l=-\frac95\imath-\frac{18}{35}\eta\,\imath=l_1$ is a good choice
(which yields a linkage that is well
suited for visualization, see Figure~\ref{figure:elliptic_linkage}).
Applying the flip procedure iteratively to this input data, one obtains the following:
\begin{alignat*}3
  \tilde{k}_1 &=  \imath - \tfrac{13}{28} \eta \, \imath, &
  \tilde{k}_2 &=  \imath + \tfrac{5}{8} \eta \, \imath, &
  \tilde{k}_3 &= -\imath - \tfrac{11}{56} \eta \, \imath, \\
  l_2 &= -\tfrac{9}{5} \imath + \tfrac{9}{20} \eta \, \imath,\quad &
  l_3 &= -\tfrac{9}{5} \imath - \tfrac{27}{40} \eta \, \imath,\quad &
  l_4 &= -\tfrac{9}{5} \imath - \tfrac{207}{140} \eta \, \imath.
\end{alignat*}

Now we can use these quantities to construct a linkage --- whose graph is
shown in Figure~\ref{figure:elliptic_graph} --- that draws an ellipse. Note
that all linear polynomials appearing here are purely imaginary in
their secondary parts. This implies that their fixed points are of the
form $(x,0)$, so that all joints are located on the horizontal axis when the
linkage is in its initial position ($t=\infty$); the same happens for $t=0$.
\end{example}

\begin{figure}[ht]
\begin{tikzpicture}[thick]
  \Vertex[x=0,y=2]{1}
  \Vertex[x=0,y=0]{5}
  \Vertex[x=3,y=2]{2}
  \Vertex[x=3,y=0]{6}
  \Vertex[x=6,y=2]{3}
  \Vertex[x=6,y=0]{7}
  \Vertex[x=9,y=2]{4}
  \Vertex[x=9,y=0]{8}
  \Edge[label=$\midpt{k_1}$](1)(2)
  \Edge[label=$\midpt{k_2}$](2)(3)
  \Edge[label=$\midpt{k_3}$](3)(4)
  \Edge[label=$\midpt{l_1}$](1)(5)
  \Edge[label=$\midpt{l_2}$](2)(6)
  \Edge[label=$\midpt{l_3}$](3)(7)
  \Edge[label=$\midpt{l_4}$](4)(8)
  \Edge[label=$\midpt{\tilde{k}_1}$](5)(6)
  \Edge[label=$\midpt{\tilde{k}_2}$](6)(7)
  \Edge[label=$\midpt{\tilde{k}_3}$](7)(8)
\end{tikzpicture}
\caption{The link graph for the linkage drawing an ellipse.}
\label{figure:elliptic_graph}
\end{figure}
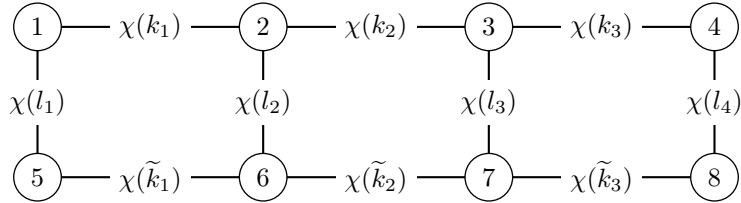

\section{Self-collisions}
\label{collisions}

A natural question that arises in the construction of a linkage is whether it
can be physically realized such that no collisions between components of the
linkage occur. More precisely, one asks whether there exists an assignment of
$1,2,\dots,n$ to the links such that for any $i<k<j$ for which $i$ and $j$ are
neighboring links, the joint connecting $i$ and $j$ never overlaps with the
link~$k$. For general linkages it is a difficult problem to detect
such collisions; see for example~\cite[Section~9.3
and Theorem~9.5.5]{GoodmanORourke2004}. However, it turns out that the same
problem has a straightforward solution for linkages constructed by our
algorithm, and the reason is that all joints in our case follow rational
curves whose explicit parametrization is a direct byproduct of the algorithm
\texttt{ConstructStrongLinkage}. Assuming that a link is realized as a
collection of line segments connecting the (two or three) joints attached to
it, a collision is described as follows: for some $i<k<j$ as above there
exists a $t\in\R\cup\{\infty\}$ such that the position
$\bigl(x_1(t),y_1(t)\bigr)$ of the joint $(i,j)$ lies on one of the line
segments of link~$k$; denote its endpoints by $\bigl(x_2(t),y_2(t)\bigr)$ and
$\bigl(x_3(t),y_3(t)\bigr)$. Then a collision between joint $\{i,j\}$ and link~$k$
happens if and only if the system
\begin{equation}\label{eq:coll} \left\{
\begin{aligned}
  x_1(t) &= s\cdot x_2(t) + (1-s)\cdot x_3(t)\\
  y_1(t) &= s\cdot y_2(t) + (1-s)\cdot y_3(t)
\end{aligned} \right.
\end{equation}
has a solution for $0\leq s\leq1$ and $t\in\R\cup\{\infty\}$. In our
construction, the coordinates~$x_i,y_i$ are rational functions in~$t$, so that
the system~\eqref{eq:coll} is equivalent to a bivariate polynomial system that
can be easily solved. Note that we have to solve \eqref{eq:coll} for each
admissible triple $i<k<j$. This way we can decide whether collisions occur or
not, and if so, we get a precise description when and where they happen.

\begin{example}
\label{example:elliptic_motion_collisions}
Consider again the linkage constructed in
Example~\ref{example:elliptic_motion_flips}, whose link graph is shown in
Figure~\ref{figure:elliptic_graph}. With the ordering $(5, 1, 6, 2, 7, 8, 4, 3)$
of the links, we get two collisions, which both happen at $t=\infty$. Hence
in principle we can trace, without disassembling the linkage, the full ellipse
except a single point; see also the animations provided on our
webpage~\cite{Koutschan15a}.
\end{example}

At the same time, we want to mention a remarkable method for preventing
collisions for planar linkages with a ladder-shaped link diagram as shown in
Figure~\ref{figure:strong_linkage}. This method is based on designing the
shape of some links and joints: we introduce three types of links (F-link,
U-link, Z-link), and two types of joints (T-joint, Z-joint).

As before, all links are arranged in different layers, and w.l.o.g.~we assume
that these layers correspond to integer numbers. An F-link is located in a
single layer, in other words, we associate one integer for each F-link. On the
other hand, U-links and Z-links stretch across two, not necessarily
neighboring, layers.  Therefore, to such links we associate a pair of integers
$(a,b)$ with $a<b$, which means that one part of the link is located in
layer~$a$ and the other part in layer~$b$; the two parts are rigidly connected
by a vertical rod. For two U-links with layers $(a_1,b_1)\neq(a_2,b_2)$,
we prohibit the situations $a_1<a_2<b_1<b_2$ and
$a_2<a_1<b_2<b_1$, which could yield a collision between these two
U-links. For a Z-link on layers $(a,b)$ we impose the condition $b-a=2$, and
the link located on layer $a+1$ (this will always be some part of
a U-link) is connected with this Z-link by a revolute joint (which we call
Z-joint) around its vertical rod. In contrast, a T-joint joins two links
located on neighboring layers. These types of links and joints are
illustrated in Figure~\ref{figure:linktypes}.

Using the above design, one can check that collisions can only happen between
links, not between links and joints. We argue now that even such
collisions can be avoided. First, an F-link cannot collide with another
F-link or a Z-link. Second, a Z-link cannot collide with another
Z-link. Third, if we move the vertical rods of all U-links sufficiently
far away, they do not collide with the F-links and Z-links. Fourth, the
above conditions on the layers of U-links imply that two U-links can collide
only if they occur in a nested way (e.g., $a_1<a_2<b_2<b_1$).  Again, this can
be avoided by moving the vertical rod of the outer U-link far enough away.
Thus we can always make a collision-free design, by manipulating the shape
of the U-links.

\begin{figure}[ht]
  \includegraphics[width=0.6\textwidth]{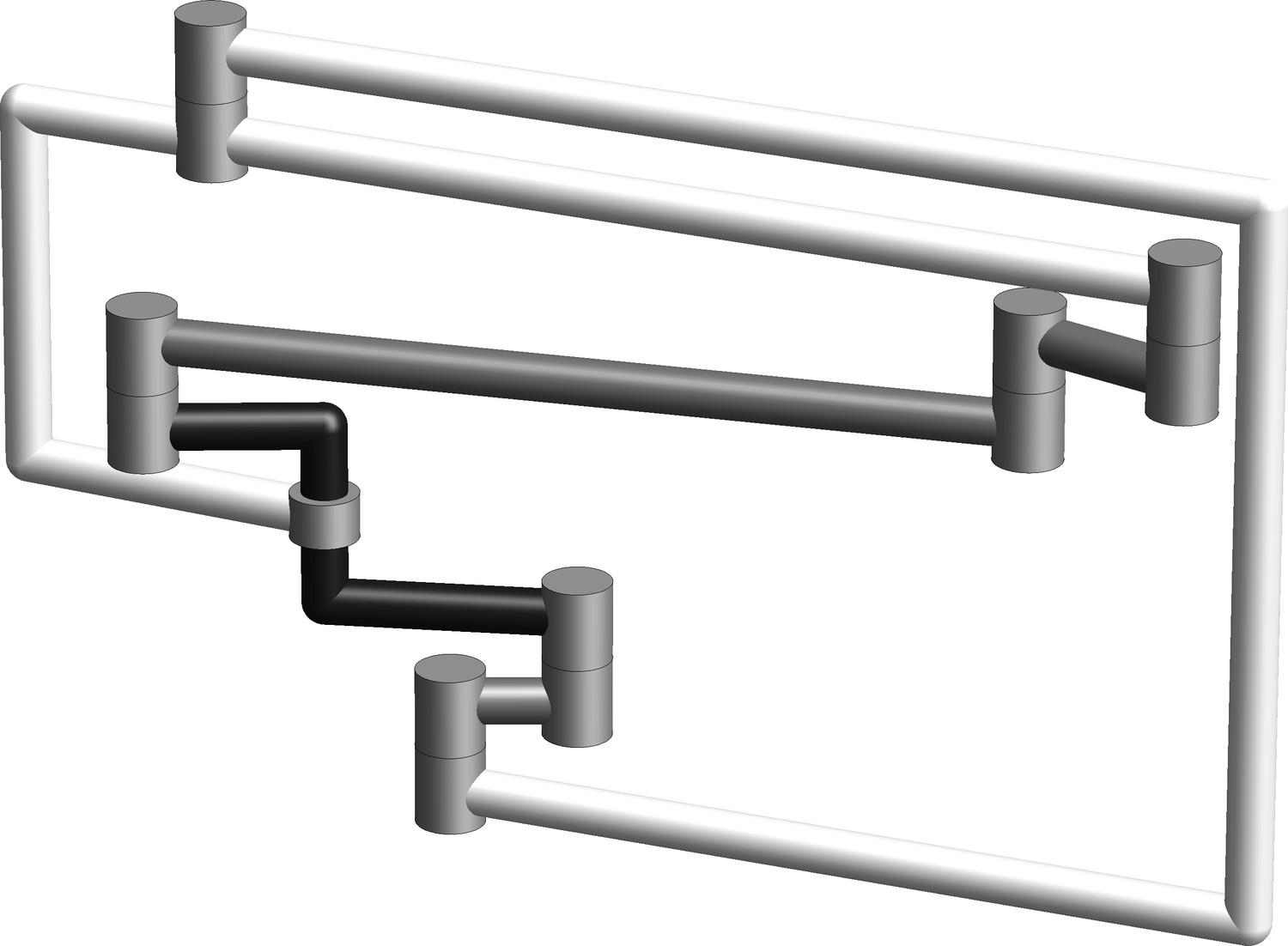}
  \caption{A linkage without self-collisions composed of three F-links (gray),
    two U-links (white), and one Z-link (black)}
  \label{figure:linktypes}
\end{figure}

It remains to argue that for a general ladder-shaped link diagram, as returned
by our algorithm \texttt{ConstructStrongLinkage} and as displayed in
Figure~\ref{figure:strong_linkage}, we can find an assignment of link types
and layers such that all the above conditions are fulfilled. One such
assignment is depicted in Figure~\ref{figure:design_linkage}; it demonstrates
that we can realize a motion --- given by a factored motion polynomial of
degree~$n$ --- by a collision-free linkage using $4n+1$ layers.

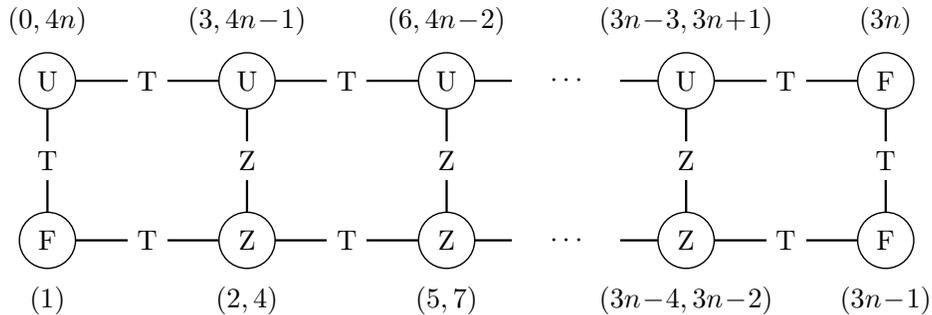
\begin{figure}[ht]\centering
 \resizebox{1\textwidth}{!}{
	\begin{tikzpicture}[thick]
	\SetGraphUnit{2}
	\begin{scope}[VertexStyle/.append style = {minimum size = 20pt}]
		\Vertex[L=U,x=0,y=2]{0}
		\node  at (0,2.75) {$(0,4n)$};
		\Vertex[L=F,x=0,y=0]{1}
		\node  at (0,-0.75) {$(1)$};
		\Vertex[L=U,x=2.5,y=2]{3}
		\node at (2.5,2.75) {$(3,4n\!-\!1)$};
		\Vertex[L=Z,x=2.5,y=0]{2}
		\node at (2.5,-0.75) {$(2,4)$};
		\Vertex[L=U,x=5,y=2]{6}
		\node at (5,2.75) {$(6,4n\!-\!2)$};
		\Vertex[L=Z,x=5,y=0]{5}
		\node at (5,-0.75) {$(5,7)$};
		\Vertex[L=U,x=8,y=2]{3n-3}
		\node at (8,2.75) {$(3n\!-\!3,3n\!+\!1)$};
		\Vertex[L=Z,x=8,y=0]{3n-4}
		\node at (8,-0.75) {$(3n\!-\!4,3n\!-\!2)$};
		\Vertex[L=F,x=10.5,y=2]{3n}
		\node at (10.5,2.75) {$(3n)$};
		\Vertex[L=F,x=10.5,y=0]{3n-1}
		\node at (10.5,-0.75) {$(3n\!-\!1)$};

		\Edge[label=T](0)(1)
		\Edge[label=T](3)(0)
		\Edge[label=T](1)(2)
		\Edge[label=Z](3)(2)
		\Edge[label=T](5)(2)
		\Edge[label=Z](6)(5)
		\Edge[label=T](3)(6)
		\Edge[label=Z](3n-3)(3n-4)
		\Edge[label=T](3n)(3n-3)
		\Edge[label=T](3n)(3n-1)
		\Edge[label=T](3n-4)(3n-1)

		\Edge[label=$\quad\cdots\quad$](6)(3n-3)
		\Edge[label=$\quad\cdots\quad$](5)(3n-4)

	\end{scope}
	\end{tikzpicture}}
      \caption{Assignment of joint types (T, Z), link types (F, U, Z) and
        layers (above and below the corresponding vertices) for realizing a
        linkage with ladder-shaped link graph without self-collisions.}
	\label{figure:design_linkage}
\end{figure}

\section{A signing linkage}
\label{examples}

We conclude with an example in connection to a popular formulation of Kempe's 
Theorem, stating that \emph{``There is a linkage that signs your name''}. 
King~\cite[Corollary 1.3]{King1999} attributes this formulation to William 
Thurston. However, as remarked by O'Rourke~\cite{ORourke2011}, it is 
very implausible that a concrete ``signing linkage'' has ever been constructed 
due to the complexity, in terms of links and joints, of the linkages produced 
following Kempe's procedure. As an example to support his claim, O'Rourke points
out that already constructing a linkage drawing the ``J'' of John Hancock's
famous signature (see Figure~\ref{figure:Hancock}) would be very difficult. 
\begin{figure}[t]
  \includegraphics[width=0.9\textwidth]{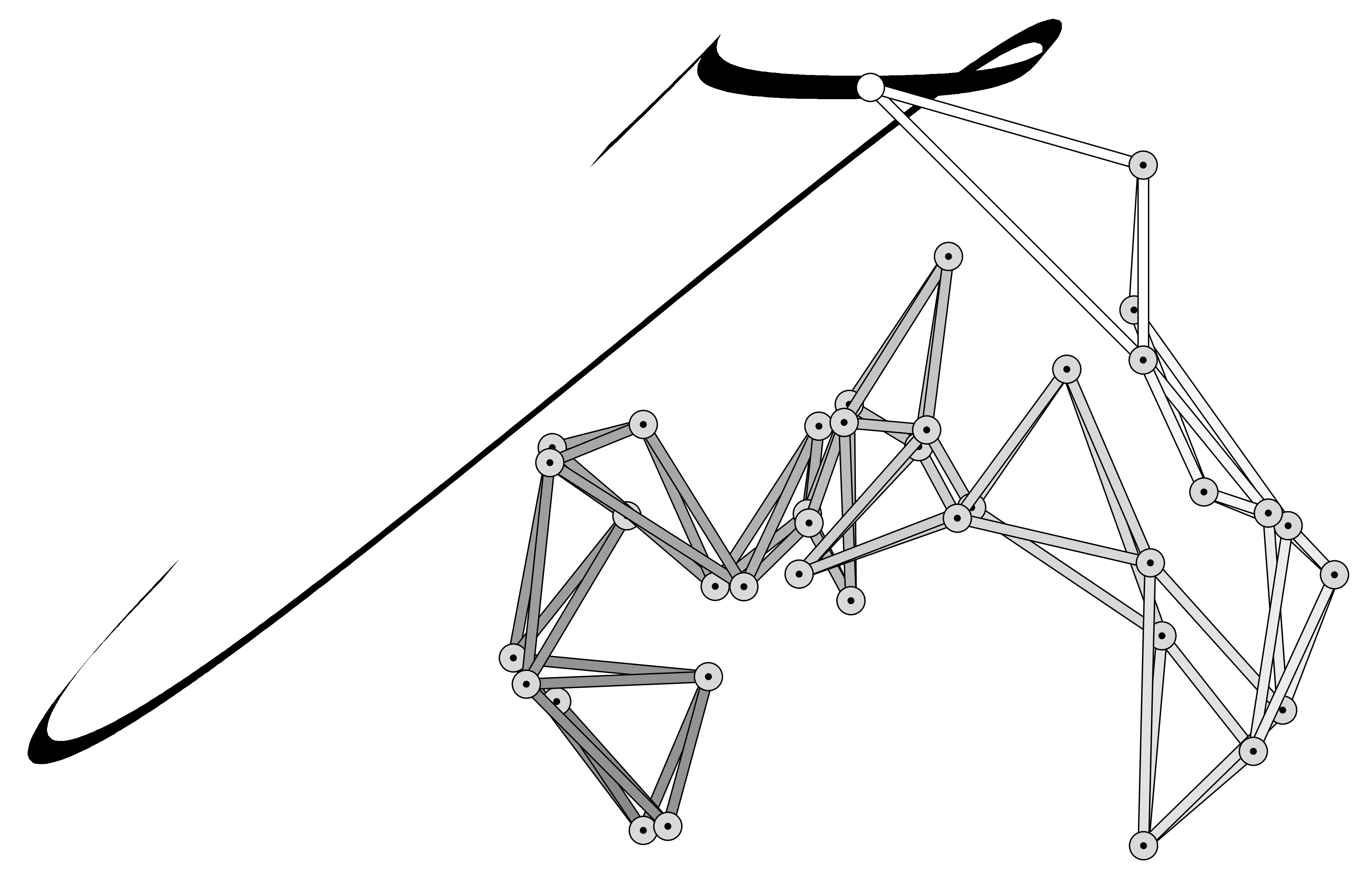}
  \caption{A rational curve approximating the ``J'' in John Hancock's
    signature and a linkage drawing it. The corresponding motion is a
    translation, which we visualize by using a quill pen whose shape is a line
    segment in direction $(5,6)$. The spatial arrangement of the links is
    indicated by hue: darker links lie below brighter ones.}
  \label{figure:J}
\end{figure}

We approximate the ``J'' by the rational curve given by the parametrization  
$(f/h,g/h)$, where 
\begin{align*}
\label{equation:Jcurve}
  f(t)&=-321880 t^5-436132 t^4-237449 t^3-64488 t^2-8666 t-451, \\
  g(t)&=-336018 t^5-472949 t^4-270569 t^3-78158 t^2-11325 t-651,\\ 
  h(t)&=170 \left(7225 t^6+13770 t^5+11187 t^4+4908 t^3+1219 t^2+162 t+9\right).
\end{align*}
Our implementation~\cite{Koutschan15a} of the algorithm~\texttt{ConstructStrongLinkage}
produces a linkage with 26 links and 37
joints realizing a translational motion along this curve, see
Figure~\ref{figure:J}; notice that these numbers are in accordance with
Proposition~\ref{prop:realize}. When this linkage draws
the depicted ``J'', one encounters seven collisions.
\begin{figure}[h]
	\includegraphics[width=0.7\textwidth]{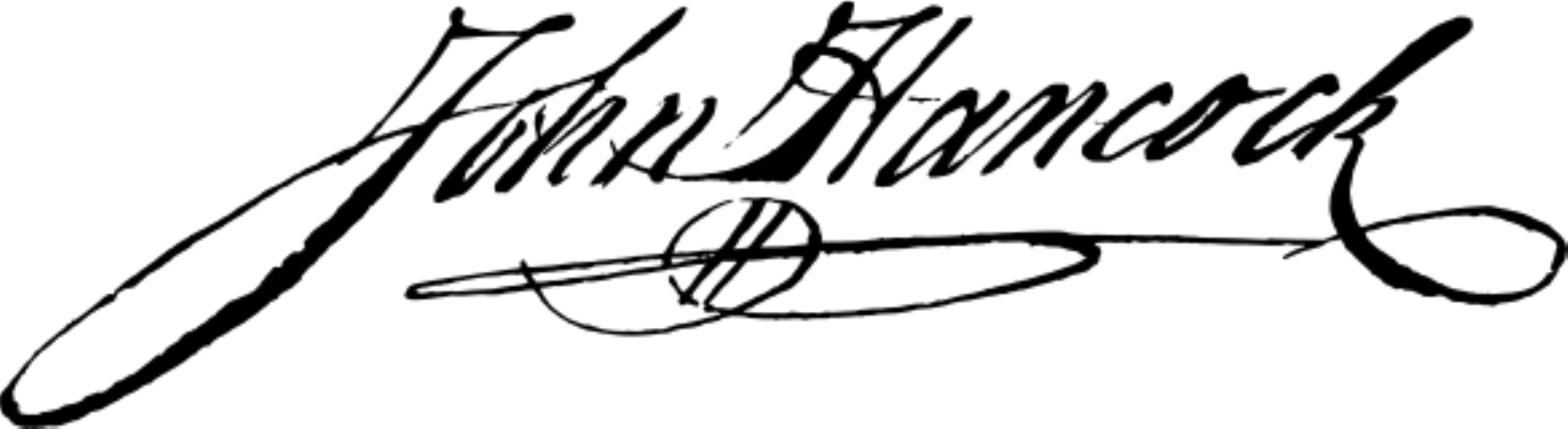}
	\caption{John Hancock's signature on the United States Declaration of 
Independence.}
	\label{figure:Hancock}
\end{figure}

\section*{Acknowledgment}
\noindent We would like to thank the anonymous referee for helpful comments.
%\newpage

\bibliographystyle{plain}
\bibliography{motion}

\begin{thebibliography}{10}

\bibitem{Abbott2008}
Timothy~G. Abbott.
\newblock Generalizations of {K}empe's universality theorem.
\newblock Master's thesis, Massachusetts Institute of Technology, 2008.
\newblock Available at \url{http://web.mit.edu/tabbott/www/papers/mthesis.pdf}.

\bibitem{BlaschkeMueller1956}
Wilhelm Blaschke and Hans~R. M{\"u}ller.
\newblock {\em Ebene {K}inematik}.
\newblock Verlag von R. Oldenbourg, M\"unchen, 1956.

\bibitem{Bochnak1998}
Jacek Bochnak, Michel Coste, and Marie-Fran{\c{c}}oise Roy.
\newblock {\em Real algebraic geometry}, volume~36 of {\em Ergebnisse der
  Mathematik und ihrer Grenzgebiete (3)}.
\newblock Springer-Verlag, Berlin, 1998.

\bibitem{Bottema1979}
Oene Bottema and Bernard Roth.
\newblock {\em Theoretical Kinematics}.
\newblock North-Holland Publishing Company, 1979.

\bibitem{GoodmanORourke2004}
Robert Connelly and Erik~D. Demaine.
\newblock Geometry and topology of polygonal linkages.
\newblock In Jacob~E. Goodman and Joseph O'Rourke, editors, {\em Handbook of
  discrete and computational geometry}, pages 197--218. Chapman \& Hall/CRC,
  Boca Raton, FL, second edition, 2004.

\bibitem{DemaineORourke2007}
Erik~D. Demaine and Joseph O'Rourke.
\newblock {\em Geometric folding algorithms}.
\newblock Cambridge University Press, Cambridge, 2007.

\bibitem{GaoZhu1999}
Xiao-Shan Gao and Chang-Cai Zhu.
\newblock Automated generation of {K}empe linkage and its complexity.
\newblock {\em Journal of Computer Science and Technology}, 14(5):460--467,
  1999.

\bibitem{GaoZhuChouGe2001}
Xiao-Shan Gao, Chang-Cai Zhu, Shang-Ching Chou, and Jian-Xin Ge.
\newblock Automated generation of {K}empe linkages for algebraic curves and
  surfaces.
\newblock {\em Mechanism and Machine Theory}, 36(9):1019--1033, 2001.

\bibitem{HegedusSchichoSchroecker2012}
G\'abor Heged\"us, Josef Schicho, and Hans-Peter Schr\"ocker.
\newblock Construction of {O}verconstrained {L}inkages by {F}actorization of
  {R}ational {M}otions.
\newblock In Jadran Lenar\v{c}i\v{c} and Manfred Husty, editors, {\em Latest
  Advances in Robot Kinematics}, pages 213--220. Springer Netherlands, 2012.

\bibitem{HegedusSchichoSchroecker2013a}
{G}\'abor {H}eged\"us, {J}osef {S}chicho, and Hans-Peter {S}chr\"ocker.
\newblock Factorization of {R}ational {C}urves in the {S}tudy {Q}uadric.
\newblock {\em {M}echanism and {M}achine {T}heory}, 69:142--152, 2013.

\bibitem{HustySchroecker}
Manfred~L. Husty and Hans-Peter Schr\"{o}cker.
\newblock Kinematics and algebraic geometry.
\newblock In J.~Michael McCarthy, editor, {\em 21st Century Kinematics}, pages
  85--123. Springer London, 2013.

\bibitem{JordanSteiner1999}
Denis Jordan and Marcel Steiner.
\newblock Configuration spaces of mechanical linkages.
\newblock {\em Discrete \& Computational Geometry}, 22(2):297--315, 1999.

\bibitem{KapovichMillson2002}
Michael Kapovich and John~J. Millson.
\newblock Universality theorems for configuration spaces of planar linkages.
\newblock {\em Topology}, 41(6):1051--1107, 2002.

\bibitem{Kempe}
Alfred~B. Kempe.
\newblock On a {G}eneral {M}ethod of describing {P}lane {C}urves of the nth
  degree by {L}inkwork.
\newblock {\em Proceedings of the London Mathematical Society}, S1-7(1):213,
  1876.

\bibitem{King1999}
Henry~C. {King}.
\newblock {Planar linkages and algebraic sets.}
\newblock {\em {Turkish Journal of Mathematics}}, 23(1):33--56, 1999.

\bibitem{Kobel}
Alexander Kobel.
\newblock Automated generation of {K}empe linkages for algebraic curves in a
  dynamic geometry system.
\newblock Bachelor's thesis, University of Saarbr{\"u}cken, 2008.
\newblock Available at
  \url{https://people.mpi-inf.mpg.de/~akobel/publications/Kobel08-kempe-linkages.pdf}.

\bibitem{Koutschan15a}
Christoph Koutschan.
\newblock Mathematica package {P}lanar{L}inkages.m and electronic supplementary
  material for the paper ``{P}lanar linkages following a prescribed motion'',
  2015.
\newblock Available at \url{http://www.koutschan.de/data/link/}.

\bibitem{Krattenthaler99}
Christian Krattenthaler.
\newblock Advanced determinant calculus.
\newblock {\em S\'{e}minaire Lothar\-ingien de Combinatoire}, 42:1--67, 1999.
\newblock Article B42q.

\bibitem{LascouxPragacz02}
Alain Lascoux and Piotr Pragacz.
\newblock Jacobians of symmetric polynomials.
\newblock {\em Annals of Combinatorics}, 6(2):169--172, 2002.

\bibitem{Lebesgue1950}
Henri Lebesgue.
\newblock {\em Le\c cons sur les {C}onstructions {G}\'eom\'etriques}.
\newblock Gauthier-Villars, Paris, 1950.

\bibitem{LiSchichoSchroecker2015}
Zijia Li, Josef Schicho, and Hans-Peter Schr\"ocker.
\newblock Factorization of motion polynomials.
\newblock {\em CoRR}, abs/1502.07600, 2015.
\newblock Available at \url{http://arxiv.org/abs/1502.07600}.

\bibitem{Malkevitch2002}
Joseph Malkevitch.
\newblock Linkages: From fingers to robot arms.
\newblock AMS Feature Column, September 2002.
\newblock Available at
  \url{http://www.ams.org/samplings/feature-column/fcarc-linkages1}.

\bibitem{ORourke2011}
Joseph O'Rourke.
\newblock {\em How to fold it}.
\newblock Cambridge University Press, Cambridge, 2011.

\bibitem{OwenPower2009}
John Owen and Stephen Power.
\newblock Continuous curves from infinite {K}empe linkages.
\newblock {\em Bulletin of the London Mathematical Society}, 41(6):1105--1111,
  2009.

\bibitem{PlecMcWamp}
Mark Plecnik, J.~Michael McCarthy, and Charles~W. Wampler.
\newblock Kinematic {S}ynthesis of a {W}att {I} {S}ix-{B}ar {L}inkage for
  {B}ody {G}uidance.
\newblock In {\em Advances in Robot Kinematics}, pages 317--325. Springer
  International Publishing, 2014.

\bibitem{Selig2005}
Jon~M. Selig.
\newblock {\em Geometric fundamentals of robotics}.
\newblock Monographs in Computer Science. Springer, New York, second edition,
  2005.

\end{thebibliography}

\end{document}